\documentclass{article}

\usepackage[margin=1in]{geometry}

\usepackage{amssymb,amsmath,amsthm,amsfonts}
\usepackage{bm}
\usepackage{textgreek}
\usepackage{mathtools}
\usepackage{enumitem}
\usepackage[numbers,comma,sort&compress]{natbib}
\usepackage{authblk}
\usepackage{graphicx}
\usepackage[font=small]{caption}
\usepackage{subcaption} 
\usepackage{float}
\usepackage[ruled,vlined,linesnumbered]{algorithm2e}

\usepackage[colorlinks]{hyperref}

\usepackage{braket}
\usepackage{float}
\usepackage{tablefootnote}
\usepackage{overpic}
\usepackage{multirow}
\usepackage{afterpage}
\usepackage{multirow}     
\usepackage{booktabs}      
\usepackage{tabularx}
\usepackage{makecell}

\newtheorem{theorem}{Theorem}
\newtheorem{definition}{Definition}
\newtheorem{lemma}{Lemma}
\newtheorem{proposition}{Proposition}

\newtheorem{assumption}{Assumption}

\newcommand{\eq}[1]{(\ref{eq:#1})}

\newcommand{\thm}[1]{\hyperref[thm:#1]{Theorem~\ref*{thm:#1}}}
\newcommand{\cor}[1]{\hyperref[cor:#1]{Corollary~\ref*{cor:#1}}}
\newcommand{\defn}[1]{\hyperref[defn:#1]{Definition~\ref*{defn:#1}}}
\newcommand{\lem}[1]{\hyperref[lem:#1]{Lemma~\ref*{lem:#1}}}
\newcommand{\prop}[1]{\hyperref[prop:#1]{Proposition~\ref*{prop:#1}}}
\newcommand{\assum}[1]{\hyperref[assum:#1]{Assumption~\ref*{assum:#1}}}
\newcommand{\fig}[1]{\hyperref[fig:#1]{Figure~\ref*{fig:#1}}}
\newcommand{\tab}[1]{\hyperref[tab:#1]{Table~\ref*{tab:#1}}}
\newcommand{\algo}[1]{\hyperref[algo:#1]{Algorithm~\ref*{algo:#1}}}
\renewcommand{\sec}[1]{\hyperref[sec:#1]{Section~\ref*{sec:#1}}}
\newcommand{\append}[1]{\hyperref[append:#1]{Appendix~\ref*{append:#1}}}
\newcommand{\fac}[1]{\hyperref[fac:#1]{Fact~\ref*{fac:#1}}}
\newcommand{\lin}[1]{\hyperref[lin:#1]{Line~\ref*{lin:#1}}}

\def\>{\rangle}
\def\<{\langle}

\newcommand{\R}{\mathbb{R}}

\DeclareMathOperator{\poly}{poly}

\title{Quantum Algorithms for Bandits with Knapsacks with Improved Regret and Time Complexities}

\author[1,2]{Yuexin Su\thanks{\texttt{yuexinsu@stu.pku.edu.cn}}}
\author[1,2,3]{Ziyi Yang\thanks{\texttt{2100010833@stu.pku.edu.cn}}}
\author[4]{Peiyuan Huang\thanks{\texttt{pyhuang@gsm.pku.edu.cn}}}
\author[1,2]{Tongyang Li\thanks{Corresponding author. \texttt{tongyangli@pku.edu.cn}}}
\author[5]{Yinyu Ye\thanks{\texttt{yinyu-ye@stanford.edu}}}

\affil[1]{Center on Frontiers of Computing Studies, Peking University}
\affil[2]{School of Computer Science, Peking University}
\affil[3]{School of Mathematical Science, Peking University}
\affil[4]{Guanghua School of Management, Peking University}
\affil[5]{Department of Management Science and Engineering, Stanford University}

\date{}

\begin{document}

\maketitle

\begin{abstract}
Bandits with knapsacks (BwK) constitute a fundamental model that combines aspects of stochastic integer programming with online learning. Classical algorithms for BwK with a time horizon $T$ achieve a problem-independent regret bound of $\mathcal{O}(\sqrt{T})$ and a problem-dependent bound of $\mathcal{O}(\log T)$. In this paper, we initiate the study of the BwK model in the setting of quantum computing, where both reward and resource consumption can be accessed via quantum oracles. 
We establish both problem-independent and problem-dependent regret bounds for quantum BwK algorithms. For the problem-independent case, we demonstrate that a quantum approach can improve the classical regret bound by a factor of $(1+\sqrt{B/\mathrm{OPT}_\mathrm{LP}})$, where $B$ is budget constraint in BwK and $\mathrm{OPT}_{\mathrm{LP}}$ denotes the optimal value of a linear programming relaxation of the BwK problem. 
For the problem-dependent setting, we develop a quantum algorithm using an inexact quantum linear programming solver. This algorithm achieves a quadratic improvement in terms of the problem-dependent parameters, as well as a polynomial speedup of time complexity on problem's dimensions compared to classical counterparts.
Compared to previous works on quantum algorithms for multi-armed bandits, our study is the first to consider bandit models with resource constraints and hence shed light on operations research.
\end{abstract}

\section{Introduction}

The multi-armed bandit (MAB) problem is a foundational model in operations research, management science, and machine learning \cite{slivkins2019introduction, gittins2011multi}. First introduced by Robbins in 1952 \cite{robbins1952some}, the MAB problem involves a decision-maker faced with a fixed set of arms, each associated with an unknown stochastic reward. During the online decision-making process, repeated selection of an arm yields increasing information about its reward distribution. The objective in the MAB problem is to minimize the cumulative regret over a fixed time horizon $T$, a goal that necessitates navigating the exploration-exploitation trade-off \cite{weber1992gittins}. 
A significant extension to this framework is the Bandits with Knapsacks (BwK) problem, which introduces resource constraints into the arm-selection process, drawing on ideas from the classical knapsack problem \cite{badanidiyuru2018bandits}. This model lies at the intersection of online linear programming (LP) and Markov decision process (MDP), the latter of which is a widely studied topic in operations research \cite{yu2009markov, borkar2002risk, wang2020randomized, bertsimas2000restless, ruszczynski2010risk}. The BwK model has also been connected to a wide range of real-world problems in operations research, including dynamic pricing and revenue management \cite{besbes2012blind, ferreira2018online}, online advertising \cite{mehta2007adwords}, and network routing \cite{agrawal2014dynamic}.

Quantum computing offers new possibilities for solving optimization problems by using quantum effects to explore complex solution spaces more efficiently than classical approaches. Techniques such as superposition and entanglement allow quantum algorithms to solve certain tasks faster. 
Within the domain of constrained optimization, quantum algorithms have demonstrated a polynomial speedup compared to their classical counterpart \cite{grigoriadis1995sublinear} in solving zero-sum games and linear programming (LP) problems through fast quantum multi-Gibbs sampling \cite{gao2024logarithmic,bouland2023quantum}. There also exist fast quantum subroutines for the simplex method~\cite{nannicini2024fast}. More general, semi-definite programming (SDP) also constitutes a significant class within constrained optimization, and quantum SDP solvers \cite{brandao2016quantum,brandao2017quantum,van2017quantum,vanApeldoorn2018SDP,augustino2023interior} offer a polynomial speedup over the current state-of-the-art classical SDP solvers \cite{lee2015faster}. Furthermore, polynomial speedup is also proved for general constrained convex optimization problems~\cite{chakrabarti2020quantum,van2020convex}. 
Beyond constrained optimization, the potential of quantum algorithms also extends to unconstrained optimization, especially nonconvex optimization \cite{gong2022robustness,liu2023quantum,zhang2023quantum,chen2025quantum,leng2023quantum,leng2025sub}. A general survey of quantum algorithms for optimization can be found at~\cite{Abbas2024}.

On the other hand, the intersection of quantum computing and reinforcement learning (RL) has yielded notable advancements \cite{dunjko2018machine,meyer2022survey,jerbi2022quantum}. For instance, Ref.~\cite{wang2021quantum} introduced a quantum exploration algorithm for efficiently identifying the optimal arm in a MAB setting, achieving a quadratic quantum speedup compared to the best possible classical result. Ref.~\cite{wan2023quantum} investigated quantum multi-armed bandits and stochastic linear bandits, establishing a regret bound of $\poly(\log T)$ regret for both models, thus achieving an exponential speedup compared to the $\Omega(\sqrt{ T} )$ regret bound of classical MAB algorithms~\cite{auer2002nonstochastic}. This $\poly(\log T)$ regret was subsequently extended to the quantum version of kernelized bandits~\cite{dai2023quantum}. More recently, Zhong et al.~\cite{zhong2023provably} studied the online exploration problem in quantum RL, presenting a novel quantum algorithm for both tabular MDPs and linear mixture MDPs with poly-logarithmic regrets.

While these theoretical developments mark significant progress, applying quantum algorithms to real-world problems—particularly those arising at the intersection of constrained optimization and operations research—remains a largely open and challenging area of research. Among practical settings in operations research, decision-making agents must simultaneously learn from data and operate under strict resource constraints, necessitating algorithmic frameworks that can handle both aspects effectively.

\paragraph{Main results}
In this paper, we investigate quantum algorithms for the Bandits with Knapsacks (BwK) problem, a core framework that integrates reinforcement learning with constrained optimization, with wide applications in operations research. Our model consists of $m$ arms and $d$ resources, operates over a finite time horizon $T$, and is subject to a total resource budget $B$ (see \sec{model} for details).
Our approach draws inspiration from classical BwK algorithms and leverages the efficiency of quantum Monte Carlo methods. Specifically, we propose two primal-dual based quantum algorithms for the BwK problem: one for the problem-independent setting where we target at a general regret bound in terms of $m$, $d$, $T$, and $B$, and the other for the problem-dependent setting where we pursue a fine-grained regret bound for each specific BwK instance (see \sec{Primal-dual} for details on problem-dependent parameters.). Our quantum algorithms are inspired by~\cite{badanidiyuru2018bandits} and~\cite{li2021symmetry}, the state-of-the-art classical algorithms for BwK in the problem-independent and problem-dependent settings, respectively. The expected regret bounds of our quantum algorithms, in comparison to their classical counterparts, are summarized in \tab{regret tab}.

\begin{table}[H]
  \centering
  \begin{tabularx}{\textwidth}{cccc}
    \toprule
    \textbf{Model} & \textbf{Method} & \textbf{Setting}  & \textbf{Regret} \\
    \midrule
    \multirow{2}{*}{\makecell{Problem\\Independent}} & \cite{badanidiyuru2018bandits} & Classical & \begin{small}
$\mathcal{O}\left(\sqrt{\log(dT)}\right)\left( \mathrm{OPT}_{\mathrm{LP}} \sqrt{\frac{m}{B}}\left(1+\sqrt{\frac{B}{\mathrm{OPT}_{\mathrm{LP}}}}\right)\right)+\mathcal{O}\left(m\log(dT)\log(T)\right)$
    \end{small} \\
    & \thm{main-independent} & Quantum & 
    \begin{small}$\mathcal{O}\left(\sqrt{\log (dT)}\right)\left(\mathrm{OPT}_{\mathrm{LP}} \sqrt{\frac{m}{B}}\right) + \mathcal{O}\left(m\log(dT)\log(T)\right)$\end{small}\\
    \midrule
    \multirow{2}{*}{\makecell{Problem\\Dependent}} & \cite{li2021symmetry} & Classical & $\mathcal{O} \left( \left( 2+\frac1b \right)^2 \frac{md}{b\delta^2}\log T + \frac{d^4}{b^2 \min \{ \chi^2,\delta^2 \} \min \{1,\sigma^2\}}\right)$ \\
    & \thm{quant-bwk2} & Quantum & $\tilde{\mathcal{O}}\left( \left(2+\frac1b\right)\frac{m \sqrt{d} }{b \delta} \log T  + \frac{d^4}{b^2 \min \{\chi^2,\delta^2\} \min\{1,\sigma^2\}} \right)$  \\
    \bottomrule
  \end{tabularx}
  \caption{Comparison of regret bounds for classical and quantum algorithms for BwK problems. (In this paper, we use the notation $\tilde{\mathcal{O}}$ to suppress poly-logarithmic factors, i.e., $\tilde{\mathcal{O}}(f)=\mathcal{O}(f\poly(\log f))$ for a function $f$.)}
    \label{tab:regret tab}
\end{table}
In the problem-independent regime, our quantum algorithm achieves an improvement in regret by a factor of $(1+\sqrt{B/\mathrm{OPT}_\mathrm{LP}})$ over the classical counterpart. 
In the problem-dependent regime, our quantum algorithm yields a quadratic improvement in the leading $\log T$ term with respect to the key parameters $\delta$ (defined in Eq.~\eq{delta}) and $(2+\frac{1}{b})$, where $b=B/T$. 

Beyond regret, our work addresses the significant computational overhead of the classical problem-dependent algorithm \cite{li2021symmetry}, which requires solving multiple computationally expensive LP problems. We mitigate this bottleneck by integrating the state-of-the-art quantum LP solver \cite{gao2024logarithmic, bouland2023quantum} into our problem-dependent quantum algorithm. This solver utilizes an improved quantum Gibbs sampling technique to accelerate LP solving, achieving a time complexity of $\tilde{\mathcal{O}}\left( \frac{\sqrt{m + d}}{\left( \epsilon_{\mathrm{LP}} / Rr \right)^{2.5} } \right)$, where $m$ is the number of variables, $d$ is the number of constraints, $R$ is an upper bound on the sum $\sum_i \xi_i$ of the optimal primal solution $\boldsymbol{\xi}$, and $r$ is an upper bound on the sum $\sum_i \eta_i$ of the optimal dual solution $\boldsymbol{\eta}$. Building on this technique, this quantum LP solver achieves a polynomial speedup in the dependence on the problem’s dimension compared to classical LP solvers  \cite{grigoriadis1995sublinear,jiang2020faster,van2021minimum,lee2015efficient} at the expense of introducing a mild polynomial term of the inverse of approximation error. To address this, we carry out a comprehensive robustness analysis that quantifies how such errors propagate through our algorithm and affect its regret guarantees, which may be of independent interest.

\begin{table}[H]
\centering
\begin{tabular}{cc}
\toprule
 \textbf{Algorithm}     &
\textbf{Time complexity} \\
\hline
\cite{li2021symmetry} 
& $
\tilde{\mathcal{O}}\left( \max\{m, d\}^{2.372} (T + m + d) \right)$ \\
\hline
\algo{quant-bwk2} 
&  $\tilde{\mathcal{O}} \left( md + \frac{(m + d)^{1.5}}{\epsilon_{\mathrm{LP}}^{2.5}} + \left( \frac{\sqrt{m + d}}{\epsilon_{\mathrm{LP}}^{2.5}} + d \right) T \right)$ \\
\bottomrule
\end{tabular}
\caption{Time complexities of classical and quantum problem-dependent algorithms for BwK.}
\label{tab:time complexity}
\end{table}

\tab{time complexity} illustrates the time complexity improvements offered by our quantum algorithm. For the classical algorithm benchmark, we note that Ref.~\cite{li2021symmetry} does not provide a detailed time complexity analysis. Since its execution requires solving LP problems to near-exact precision, we adopt the classical LP solver from Ref.~\cite{jiang2020faster}. (Different classical LP solvers exhibit varying time complexities depending on the specific parameter regime of the LP problem, as discussed in \sec{time complexity analysis} and \append{LP solver}. We note that our quantum algorithm maintains a polynomial speedup over the number of variables $m$ and constraints $d$ compared to existing classical approaches.) Notably, this solver has poly-logarithmic dependence on the inverse of precision parameter, which is hence omitted in the $\tilde{O}$ time complexity bound. In contrast, our \algo{quant-bwk2} incorporates the aforementioned quantum LP solvers from \cite{gao2024logarithmic, bouland2023quantum} and achieves polynomial improvements in time complexity with respect to both $m$ and $d$.

\paragraph{Techniques}
In our quantum algorithm, we leverage the quantum Monte Carlo method (QMC) to obtain tighter upper and lower confidence bounds (UCB and LCB) on random variables. Intuitively, after $N$ samples the confidence‐interval length in classical Monte Carlo methods scales as $\mathcal{O}({1}/{\sqrt{N}})$, whereas in the quantum Monte Carlo method it improves to $\tilde{\mathcal{O}}({1}/{N})$ via quantum amplitude estimation (see \sec{monte carlo} for details). By exploiting this faster convergence rate, we design quantum algorithms that achieve better regret bounds. However, QMC yields useful information only upon measurement—and each measurement collapses the quantum state—so we cannot reuse information from previous samples within the same quantum run, unlike in classical Monte Carlo methods. Consequently, in both our problem-independent and problem-dependent algorithms, QMC executions are scheduled at geometrically increasing time intervals for updating UCB and LCB. 

Furthermore, to accelerate our problem-dependent algorithm, we employ quantum LP solvers. While quantum LP solvers offer a polynomial speedup with respect to the number of variables and constraints, it yields approximate solutions rather than the near-optimal ones typically obtained from classical LP solvers, such as those by interior point methods. To ensure the reliability of our quantum algorithm in the presence of such approximations, we develop a novel robustness analysis framework. 
If the LP approximation error $\epsilon_{\textrm{LP}}$ is appropriately bounded (as detailed in \prop{regret 1} and \prop{regret 2}), the overall regret guarantees are largely preserved. This suggests that employing approximate LP solvers does not significantly compromise the algorithm’s performance. Consequently, our robustness analysis demonstrates that the superior regret performance characteristic of our quantum BwK algorithm is preserved despite these approximation errors.

\section{Preliminaries}
\subsection{Model and setup}
\label{sec:model}

\paragraph{Problem formulation}
We consider the BwK problem with $m$ arms, $d$ resources, and a finite time horizon $T$. In each round $t\in[T]$, the algorithm selects an arm $i \in [m]$, receives a reward $r_t \in [0,1]$, and consumes a resource vector $\boldsymbol{c}_t \in [0,1]^d$. Each resource $j \in [d]$ is subject to a budget constraint $B_j\in\R^{+}$. The algorithm terminates either when $T$ rounds have been completed or when any resource budget is exceeded.

We assume that for each arm $i\in[m]$, the reward and cost vectors are i.i.d.~drawn from an unknown distribution $\boldsymbol{\pi}_i$ supported on $[0,1]^{d+1}$. Thus, a problem instance is fully characterized by the set of distributions $(\boldsymbol{\pi}_i)_{i=1}^m$. We denote the expected reward vector as $\boldsymbol{r} \in [0,1]^m$ and the expected cost matrix as $\boldsymbol{C} \in [0,1]^{d \times m}$.
Throughout the paper, we adopt the convention that scalars are denoted by non-bold variables (e.g., $r_i, C_{j,i}$) and vectors/matrices by bold variables (e.g., $\boldsymbol{r}, \boldsymbol{C}$). In particular, $\boldsymbol{C}_{\cdot,i}$ represents the $i$-th column vector of the matrix $\boldsymbol{C}$ and $\boldsymbol{C}_{j,\cdot}$ represents the $j$-th row vector of the matrix $\boldsymbol{C}$.

The algorithm's expected total reward is defined as the expectation of the sum of received rewards over rounds. Let $\mathrm{OPT}$ denote the maximum achievable expected reward under the resource constraints. The  \emph{expected regret} of an algorithm is then defined as $\mathrm{OPT}$ minus its expected total reward. Our goal is to design algorithms that minimize the expected regret.

\paragraph{Assumptions}
Without loss of generality, we assume that all resource budgets are uniform, i.e., $B_j = B$ for all $j \in [d]$. This uniformity can be achieved by scaling each resource’s consumption vector $\boldsymbol{C}_{j,\cdot} \in [0,1]^m$ by a factor of $B_j / B$, where $B = \min_j B_j$. To ensure feasibility of the resource constraints, we require $B_j \leq T$ for all $j \in [d]$. Define $b = B/T \in [0,1]$, and let $\boldsymbol{B} = (B, B, \ldots, B)^\top$ denote the budget vector.

For analytical convenience, we incorporate the time horizon $T$ as an additional (first) resource constraint. That is, each arm consumes $b$ units of  time resource in every round, and we set $C_{1,i} = b$ for all $i \in [m]$.

\subsection{Primal-dual perspective for BwK}
\label{sec:Primal-dual}
\paragraph{LP relaxation of optimal value}
To analyze the BwK problem, we begin by formulating its LP relaxation. Recall that $\boldsymbol{r}$ denotes the vector of expected rewards, $\boldsymbol{C}$ represents the matrix of expected resource consumption, and $\boldsymbol{B}$ is the budget vector. The decision variables are represented by the vector $\boldsymbol{\xi} = (\xi_1, \xi_2, \dots, \xi_m) \in \mathbb{R}^m$. The primal LP is formulated as:
\begin{equation}
\label{eq:primal-LP}
\begin{aligned}
\mathrm{OPT}_{\mathrm{LP}}:=
\max_{\boldsymbol{\xi}}\quad &\boldsymbol{r}^\top \boldsymbol{\xi} \\
\mathrm{s.t.}\quad &\boldsymbol{C}\boldsymbol{\xi} \leq  \boldsymbol{B} \\
& \boldsymbol{\xi} \geq {0}.
\end{aligned}
\end{equation}
We denote the optimal value to this primal problem as $\mathrm{OPT}_{\mathrm{LP}}$, and the optimal solution as $\boldsymbol{\xi}^*$. By \cite[Lemma 3.1]{badanidiyuru2018bandits}, $\mathrm{OPT}_{\mathrm{LP}}$ upper bounds the true optimal reward $\mathrm{OPT}$ of the BwK problem. The gap between $\mathrm{OPT}_{\mathrm{LP}}$ and the expected reward of our algorithm therefore provides an upper bound on the expected regret.  The corresponding dual LP problem is formulated as follows:
\begin{equation}
\label{eq:dual-LP}
\begin{aligned}
\min_{{\eta}}\quad  
& \boldsymbol{B}^{\top} \boldsymbol{\eta}\\
\mathrm{s.t.}\quad & \boldsymbol{C}^{\top} \boldsymbol{\eta} \geq \boldsymbol{r}\\
&\boldsymbol{\eta} \geq {0},
\end{aligned}
\end{equation}
where $\boldsymbol{\eta} \in \mathbb{R}^d$ is the vector of dual variables, and we denote the dual optimal solution as $\boldsymbol{\eta}^*$.

\paragraph{Arm-bandit symmetry}
We introduce some fundamental properties concerning the relationship between the primal and dual LP problems. We partition the set of arms into optimal arms $\mathcal{I}^{*}$ and sub-optimal arms $\mathcal{I^\prime}$. Similarly, we divide the set of constraints into binding constraints $\mathcal{J}^{*}$ and non-binding constraints $\mathcal{J^\prime}$ according to the following equations:
\begin{equation}
\begin{aligned}
    \mathcal{I}^{*} &:= \{\xi_i^*>0,i\in[m]\} , \\
    \mathcal{I^\prime} &:= \{\xi_i^*=0,i\in[m]\} , \\
    \mathcal{J}^* &:= \left\{B-\sum_{i=1}^m C_{ji}\xi_i^*=0,\ j \in[d] \right\}, \\
    \mathcal{J^\prime} &:= \left\{B-\sum_{i=1}^m C_{ji}\xi_i^*>0,\ j \in [d] \right\}. 
\end{aligned}
\end{equation}
As defined, the optimal arms $\mathcal{I}^{*}$ are those selected in the optimal primal solution $\boldsymbol{\xi}^*$, whereas the sub-optimal arms $\mathcal{I}^{\prime}$ are not. For the binding constraints $\mathcal{J}^*$, the corresponding resources are fully utilized when the BwK algorithm terminates. Conversely, the non-binding constraints $\mathcal{J}^\prime$ have remaining capacity. This implies that the sets of optimal arms and binding constraints predominantly govern the behavior of the BwK algorithm. Furthermore, we introduce the following assumption to ensure the uniqueness and non-degeneracy of the relaxed LP problem. This assumption can be satisfied for any LP problem through an arbitrarily small perturbation \cite{megiddo1989varepsilon}.
\begin{assumption}
\label{assum:lp}
    $\mathrm{OPT}_{\mathrm{LP}}$ in \eq{primal-LP} has an unique optimal solution, and the optimal solution is non-degenerate:
    \begin{equation}
        \vert \mathcal{I}^* \vert = \vert \mathcal{J}^* \vert.
    \end{equation}
\end{assumption}

Building on these definitions, we define the LP problems $\mathrm{OPT}_i$ and $\mathrm{OPT}_j$ to further explore the properties of optimal arms and binding constraints. These problems involve removing arm $i$ or constraint $j$ from the original LP. 
\begin{equation}
\begin{aligned}
\label{eq:opt-i}
    \mathrm{OPT}_i:=
    \max_{\boldsymbol{\xi}} \quad & {\boldsymbol{r}}^\top \boldsymbol{\xi}, \\
    \mathrm{s.t.} \quad & \boldsymbol{C} \boldsymbol{\xi} \leq \boldsymbol{B}, \\
    & \xi_i=0, \boldsymbol{\xi} \geq \boldsymbol{0}.
\end{aligned}
\end{equation}
\begin{equation}
\begin{aligned}
\label{eq:opt-j}
    \mathrm{OPT}_j :=
    \min_{\boldsymbol{\eta}}\quad & \boldsymbol{B}^\top \boldsymbol{\eta} - B,\\
    \mathrm{s.t.}\quad & {\boldsymbol{C}}^\top \boldsymbol{\eta} \geq \boldsymbol{r} +\boldsymbol{C}_{j,\cdot},\\
    &\boldsymbol{\eta} \geq0 .\
\end{aligned}
\end{equation}
\prop{optimal gap} characterizes the impact of removing arms or modifying constraints in terms of optimality gaps.
\begin{proposition}[{\cite[Proposition 2]{li2021symmetry}}]
\label{prop:optimal gap}
With \assum{lp}, the following inequality holds:
\begin{equation}
\begin{aligned}
&\mathrm{OPT}_{i} < \mathrm{OPT}_\mathrm{LP} \Leftrightarrow i\in\mathcal{I}^{*}, \\
&\mathrm{OPT}_{i} = \mathrm{OPT}_\mathrm{LP} \Leftrightarrow i\in\mathcal{I'}, \\
&\mathrm{OPT}_{j} = \mathrm{OPT}_\mathrm{LP} \Leftrightarrow j\in\mathcal{J}^*, \\
&\mathrm{OPT}_{j} < \mathrm{OPT}_\mathrm{LP} \Leftrightarrow j\in\mathcal{J'}. 
\end{aligned}
\end{equation}
\end{proposition}

\paragraph{Notations for problem-dependent BwK}
To ensure consistency with prior work on problem-dependent BwK, we adopt some notations introduced in \cite{li2021symmetry}. Given prior definitions of the optimal arms $\mathcal{I}^*$ and the binding constraints $\mathcal{J}^*$, and based on \prop{optimal gap}, we define the parameter $\delta$ to quantify the gap between the optimal LP value $\mathrm{OPT}_{\mathrm{LP}}$ and the best sub-optimal LP values associated with arms and constraints:
\begin{equation}
\label{eq:delta}
    \delta:=\frac1T\left(\mathrm{OPT}_{\mathrm{LP}}-\max\left\{\max_{i\in\mathcal{I}^*}\mathrm{OPT}_i,\max_{j\in\mathcal{J}^{\prime}}\mathrm{OPT}_j\right\}\right).
\end{equation}
We further define the parameter $\sigma$ as the minimum singular value of the submatrix $\boldsymbol{C}_{\mathcal{J}^*,\mathcal{I}^*}$ of the resource consumption matrix $\boldsymbol{C}$: 
\begin{equation}
    \sigma = \sigma_{\min}\left(\boldsymbol{C}_{\mathcal{J}^*,\mathcal{I}^*} \right).
\end{equation}
Finally, $\chi$ is defined to capture the smallest non-zero component of the optimal primal solution:
\begin{equation}
    \chi=\frac{1}{T} \min\{\xi^*_i \neq 0, i \in [m] \}.
\end{equation}

\subsection{Classical and quantum oracles}
In the classical setting, pulling an arm yields a random reward sample and a corresponding $d$-dimensional resource consumption vector. To formally analyze this process, we model the act of sampling as an interaction with an ``oracle". This concept is then extended to the quantum domain. Specifically, we define both classical and quantum oracles for the reward $r_i$ and the resource consumption vector $\boldsymbol{C}_{\cdot,i}$ associated with each arm $i$.

Let $\Omega_i^r$ and $\boldsymbol{\Omega}_i^{\boldsymbol{C}} = (\Omega_{i,1}, \Omega_{i,2}, \ldots, \Omega_{i,d})$ denote the finite sample spaces of the reward and the $d$-dimensional resource consumption, respectively for arm $i$. We then define the classical oracles $f_i^r$ and $f_i^{\boldsymbol{C}}$, as well as the corresponding quantum oracles $\mathcal{O}_i^r$ and $\mathcal{O}_i^{\boldsymbol{C}}$, for querying the reward and resource consumption of arm $i$.

\begin{definition}[Classical oracles for BwK]
\label{defn:classical-oracle}
Define 
\begin{equation}
\begin{aligned}
    f_{i}^r &\colon i \rightarrow X_{i}^r(\omega),\: \omega \in \Omega_i^r \quad (\text{reward oracle}), \\
    f_{i}^{\boldsymbol{C}} &\colon  i \rightarrow X_{i}^{\boldsymbol{C}}(\omega),\: \omega \in \Omega_i^{\boldsymbol{C}} \quad(\text{resource consumption oracle}),
\end{aligned}
\end{equation}
where $X_i^r(\omega)$ is the random reward at $\omega$ and $X_i^{\boldsymbol{C}}(\omega)$ is the corresponding random $d$-dimensional resource consumption vector.
\end{definition}
In essence, the classical oracles model the process of drawing a sample from a random variable. Each query to a classical oracle for a given arm returns an immediate sample of both the reward and the resource consumption for that arm.

\begin{definition}[Quantum oracles for BwK]
\label{defn:quantum-oracle}
\begin{equation}
\begin{aligned}
\mathcal{O}_{i}^r &\colon
\ket{0} \to \sum_{\omega \in \Omega_{i}^r}
\sqrt{P(\omega)} \ket{\omega} \ket{y(\omega)} \quad(\text{reward oracle}), \\
\mathcal{O}_{i}^{\boldsymbol{C}} &\colon
\ket{0} \to \sum_{\omega \in \Omega_{i}^{\boldsymbol{C}}}
\sqrt{P(\omega)} \ket{\omega} \ket{y(\omega)} \quad(\text{resource consumption oracle}),
\end{aligned}
\end{equation}
where $y(\omega)\colon \Omega \to [0,1]$ or $y(\omega)\colon \Omega \to [0,1]^d$ is the random reward or resource consumption associated with $\omega$.
\end{definition}

In quantum computing, a single qubit is represented by a vector in the complex Hilbert space $\mathbb{C}^2$, while an $n$-qubit quantum state resides in the space $\mathbb{C}^{2^n}$. Given two quantum states $\ket{v} \in \mathbb{C}^{2^{n_1}}$ and $\ket{u} \in \mathbb{C}^{2^{n_2}}$, their tensor product $\ket{v} \otimes \ket{u}$ is a vector in $\mathbb{C}^{2^{n_1 + n_2}}$ of the form $(v_0 u_0, v_0 u_1, \ldots, v_{2^{n_1}-1} u_{2^{n_2}-1})$. For brevity, this tensor product is often denoted by $\ket{v}\ket{u}$. Quantum oracles in \defn{quantum-oracle} prepare a superposition over all possible outcomes from the underlying reward and cost distributions associated with each arm. This definition naturally generalizes the classical sampling oracle: if a measurement is performed on the state $\mathcal{O}_{i}^{r}\ket{0}$ or $\mathcal{O}_{i}^{\boldsymbol{C}}\ket{0}$, the quantum oracle  behaves like a classical oracle by returning a single sample according to the corresponding distribution. In contrast to classical oracles, which provide the exact reward and resource costs associated with each arm pull, quantum oracles provide no direct feedback until a measurement is performed.  When used within a quantum algorithm, the algorithm can query both the unitary oracle $\mathcal{O}$ and its inverse $\mathcal{O}^\dagger$ during each interaction with the arm.

Additionally, the quantum LP solver requires oracle access to the coefficients of the linear programming problem (see \defn{lp solver setting} for details). We implement this oracle by storing the coefficients in a quantum-read, classical-write random access memory (QRAM) \cite{giovannetti2008quantum}. 
Specifically, QRAM allows for the storage and modification of an array $a_1,a_2,\dots,a_l$ of classical data while providing quantum read-only access. That is, it enables the implementation of a unitary operator $\mathcal{U}_{\textrm{QRAM}}$ such that
\begin{equation}
    \mathcal{U}_{\textrm{QRAM}}\colon \ket{i} \ket{0} \rightarrow \ket{i} \ket{a_i}.
\end{equation}
Given $l$ classical data points, the QRAM can be constructed in $\tilde{\mathcal{O}}(l)$ time and allows access to each element in $\tilde{\mathcal{O}}(1)$ time \cite{kerenidis2016quantum,grover2002creating}.

\subsection{Classical and quantum UCB and LCB}
\label{sec:monte carlo}
Classical upper confidence bound (UCB) and lower confidence bound (LCB) algorithms are known to achieve the following guarantee by the Chernoff-Hoeffding's inequality.
\begin{lemma}[Classical univariate and multivariate mean estimator]
\label{lem:classic-monto}
Let $\boldsymbol{y} \colon \Omega \to [0,1]^d$ be a $d$-dimensional random variable with bounded value, and let $\hat{\boldsymbol{y}}$ be the average of $N$ independent samples from this distribution. Then the Chernoff-Hoeffding's inequality guarantees that $P(\vert \hat{\boldsymbol{y}}- \mathbb{E}[\boldsymbol{y}] \vert \leq \epsilon) \geq 1- \delta $ with at most $\frac{\log\left( 2d/\delta \right)}{\epsilon^2}$ samples.
\end{lemma}

In the quantum setting, by employing quantum oracles and quantum Monte Carlo (QMC) techniques, we achieve a nearly quadratic enhancement in estimation precision. This acceleration stems from quantum amplitude estimation \cite{brassard2000quantum}, as detailed in the following discussion. 
Suppose we have a quantum algorithm $\mathcal{A}$ acting on $n$ qubits, whose output state is
\begin{equation}
    \ket{\phi}:=\mathcal{A}\ket{0}^{\otimes n} = \sum_x \sqrt{p(x)} \ket{x},
\end{equation}
where $\{\ket{x} \mid x \in \{0,1\}^n\}$ denotes the set of $2^n$ computational basis states, and $p(x)\in[0,1]$ is the probability of obtaining outcome $x$ upon measurement. If the measured outcome is $x$, the algorithm returns a value $\phi(x)$, and our goal is to estimate the expectation $\mathbb{E}[\phi]=\sum_x p(x)\phi(x)$. To achieve this, the QMC employs quantum amplitude estimation. Specifically, we introduce one ancillary qubit and apply a unitary operator $W$ defined by 
\begin{equation}
    \ket{x}\ket{0} \rightarrow \ket{x} (\sqrt{\phi(x)}\ket{1} + \sqrt{1-\phi(x)}\ket{0} ).
\end{equation} 
Consequently, applying $\mathcal{A}$ on the first $n$ qubits followed by $W$ on the appended qubit transforms the initial state $\ket{0}^{\otimes(n+1)}$ into 
\begin{equation}
\begin{aligned}
   \ket{\psi} := W(\mathcal{A} \otimes I)\ket{0}^{\otimes(n+1)} &= \sum_x \left(\sqrt{p(x)\phi(x)} \ket{x}\ket{1} + \sqrt{p(x)(1-\phi(x))}\ket{x}\ket{0}\right) \\
   &=\sqrt{\mathbb{E}[\phi]}\ket{\psi_1}\ket{1} + \sqrt{1-\mathbb{E}[\phi]}\ket{\psi_0}\ket{0},
\end{aligned}
\end{equation}
where $\ket{\psi_1}:=\frac{1}{\sqrt{\mathbb{E}[\phi]}}\sum_x\sqrt{p(x)\phi(x)} \ket{x}$ and $\ket{\psi_0}:=\frac{1}{\sqrt{1-\mathbb{E}[\phi]}} \sum_x\sqrt{p(x)(1-\phi(x))} \ket{x}$.
In this form, the probability of measuring the ancilla in $\ket{1}$ is exactly $\mathbb{E}[\phi]$.  
Amplitude estimation then proceeds by alternating reflections about $\ket{\psi}$ and about the bad subspace $\ket{\psi_0}\ket{0}$ through the reflection on the last qubit. After $t$ such iterations, one obtains an estimate $\tilde{\phi}$ satisfying $\vert \tilde{\phi} - \mathbb{E}[\phi] \vert \leq \mathcal{O}(\sqrt{\mathbb{E}[\phi]}/t + 1/{t^2})$~\cite{brassard2000quantum}. This can be further extended to fast quantum univariate and multivariate mean estimators as follows:

\begin{lemma}[Quantum univariate mean estimator {\cite{montanaro2015quantum,kothari2023mean}}] 
\label{lem:quant-uni}
Assume that $y\colon \Omega \to [0,1]$ is a random variable, $\Omega$ is equipped with a probability measure $P$, and the quantum unitary oracle $\mathcal{O}$ encodes $P$ and $y$.
There is a constant $C_1 \geq 1$ and a quantum algorithm $\mathrm{QMC}_1(\mathcal{O},\epsilon,\delta)$ which returns an estimate $\hat{y}$ of $\mathbb{E}[y]$ such that $P(\vert \hat{y}- \mathbb{E}[y] \vert \leq \epsilon) \geq 1-\delta$ using at most $\frac{C_1}{\epsilon} \log \frac1\delta$ queries to $\mathcal{O}$ and $\mathcal{O}^\dagger $.
\end{lemma}

\begin{lemma}[Quantum multivariate mean estimator {\cite{cornelissen2022near}}] 
\label{lem:quant-multi}
Assume that $ \boldsymbol{y}\colon \Omega \to [0,1]^d$ is a $d$-dimensional random variable, $\Omega$ is equipped with a probability measure $P$, and the quantum unitary oracle $\mathcal{O}$ encodes $P$ and $\boldsymbol{y}$.
There is a constant $C_2 \geq 1$ and a quantum algorithm $\mathrm{QMC}_2(\mathcal{O},\epsilon,\delta)$ which returns an estimate $\hat{\boldsymbol{y}}$ of $\mathbb{E}[\boldsymbol{y}]$ such that $P(\Vert \hat{\boldsymbol{y}}- \mathbb{E}[\boldsymbol{y}] \Vert_\infty \leq \epsilon) \geq 1-\delta$ using at most $ \frac{C_2 \sqrt{d} \log(d/\delta) }{\epsilon} \sqrt{\log\left(\frac{\sqrt{d} \log(d/\delta) }{\epsilon} \right)} $ queries to $\mathcal{O}$ and $\mathcal{O}^\dagger $.
\end{lemma}
As established in \lem{quant-uni} and \lem{quant-multi}, quantum Monte Carlo reduces the confidence‐interval length to $\tilde{\mathcal{O}}(1/N )$ when an arm is sampled $N$ times. By contrast, any classical algorithm that samples an arm $N$ times attains only $1/{\sqrt{N}}$ scaling. This nearly quadratic improvement in the confidence interval width is a key advantage of the quantum approach compared to classical techniques.

\section{Problem-Independent Quantum BwK}
In this section, we introduce a problem-independent quantum algorithm for the BwK problem.  Our approach builds upon the classical BwK algorithm proposed in Ref.~\cite{badanidiyuru2018bandits}, with the key modification of replacing classical oracles with quantum oracles. 
By leveraging quantum Monte Carlo methods, as stated in \lem{quant-uni}, our algorithm achieves a quadratic improvement in the estimation of confidence intervals, thereby enhancing the performance of the classical PrimalDualBwK algorithm.  

We present the pseudocode of our quantum problem-independent algorithm in \algo{quant-bwk1}. The algorithm utilizes quantum oracles for both reward and resource consumption for each arm $i$, as defined in \defn{quantum-oracle}. Each time the quantum reward oracle for arm $i$ is queried, the immediate expected reward $r_i$ is added to the total expected reward. To ensure that the resource constraint is not exceeded, the quantum resource consumption oracle is queried using an immediate measurement, acting as a classical oracle. As a result, the quantum Monte Carlo method is employed solely to estimate the UCB and LCB of rewards, while the UCB and LCB for resource consumption are computed classically.
During the execution of our algorithm, each arm is initially played for $N$ consecutive rounds to obtain preliminary estimates of UCB and LCB for both rewards and resource consumption vectors. The algorithm then proceeds by selecting the arm that maximizes the ``bang-per-buck" ratio, defined as the expected reward divided by the expected cost.  
Due to the fact that classical estimators allow per-round updates, whereas quantum estimators can only update the UCB and LCB after multiple oracle queries due to entanglement, the confidence bounds for resource consumption are updated classically after each pull of the selected arm, while the quantum confidence bounds for rewards are updated only after the arm has been pulled multiple times. We use $\boldsymbol{r}^U(s)$ and $\boldsymbol{C}^L(s)$ to represent the UCB of rewards and LCB of resource consumption vector at time $s$.

\vspace{4mm}
\begin{algorithm}[!htbp]
    \SetAlgoLined
    \caption{Problem-Independent Quantum BwK}
    \label{algo:quant-bwk1}
    
    \Begin{
    \tcp{Phase \uppercase\expandafter{\romannumeral1}: Initialization}
    \For{$i=1 \to m$}{
        $\mathrm{rad}_i \leftarrow 1$ and $N \leftarrow \frac{2C_1}{\mathrm{rad}_i} \log T$ \;
        Play arm $i$ for consecutive $N$ rounds \;
        Run QMC to get UCB $r^U_i(1) \in [0,1]$ of the reward and run classical algorithm to get LCB $\boldsymbol{C}^L_{\cdot,i}(1)$ for resource consumption\;
        $\mathrm{rad}_{i} \leftarrow \mathrm{rad}_{i}/2$\;
    }
    $\boldsymbol{v}(1) = \boldsymbol{1}_d $\;

    \tcp{Phase \uppercase\expandafter{\romannumeral2}: pick arms that maximize the "bang-per-buck" ratio}
    \For{time $s=1,2,\dots,\tau$}{

    Let $i(s) \leftarrow \arg\max_i r^U_i (s)/(\boldsymbol{C}^L_{\cdot,i} (s) \cdot \boldsymbol{v}(s)) $ \; 
    Play arm $i(s)$ once and
    get $\boldsymbol{C}^L_{\cdot,i(s)}(s+1)$ classically \;

    \If{ arm $i(s)$ has been played for $\frac{2C_1}{\mathrm{rad}_{i(s)}} \log T $ times since the last execution of QMC}
    {Run QMC to update $r^U_i (s)$\;
    $\mathrm{rad}_{i(s)} \leftarrow \mathrm{rad}_{i(s)}/2$ \;}
    
    Update $\boldsymbol{v}(s)$:
    \begin{equation*}
        v(s+1)_j = v(s)_j (1+\epsilon)^{C^L_{j,i(s)}(s)}
    \end{equation*}
    }
    }
\end{algorithm}
\vspace{4mm}

\begin{theorem}[Main theorem, problem-independent setting]\label{thm:main-independent}
    The expected regret of \algo{quant-bwk1} satisfies
    \begin{equation}
        \mathrm{OPT}_{\mathrm{LP}}-\mathrm{REW}\leq 
        \mathcal{O}\left(\sqrt{\log (dT)}\right)\left(\mathrm{OPT}_{\mathrm{LP}} \sqrt{\frac{m}{B}}\right) + \mathcal{O}\left(m\log(dT)\log T\right).
    \end{equation}
\end{theorem}

The proof of \thm{main-independent} can be found in \append{proof in algo1}. In contrast to its classical counterpart, which establishes a regret bound of
\begin{equation}
\mathcal{O}\left(\sqrt{\log(dT)}\right)\left( \mathrm{OPT}_{\mathrm{LP}} \sqrt{\frac{m}{B}}\left(1+\sqrt{\frac{B}{\mathrm{OPT}_{\mathrm{LP}}}}\right)\right)+\mathcal{O}\left(m\log(dT)\log(T)\right).
\end{equation}
Our result in \thm{main-independent} achieves a tighter upper bound by a factor of $(1+\sqrt{B/\mathrm{OPT}_\mathrm{LP}})$. 
This indicates the advantage of quantum algorithms for solving the BwK problem. The analysis of the quantum time complexity is straightforward. The time complexity is determined by the quantum Monte Carlo method used for reward estimation. Specifically, a QMC estimation that makes $N$ queries to the oracle has a time complexity of $\tilde{\mathcal{O}}(N)$. Since the algorithm utilizes at most $T$ oracle queries in total, its overall quantum time complexity is $\tilde{\mathcal{O}}(T)$.

\section{Problem-Dependent Quantum BwK}
\subsection{Algorithm}
In this section, we introduce a problem-dependent quantum algorithm for the BwK problem, leveraging the primal-dual approach. Our algorithm builds upon the classical problem-dependent BwK algorithm proposed in \cite{li2021symmetry}. The key modifications in our approach exploit the quantum advantages of Monte Carlo methods, as outlined in \sec{monte carlo}, along with quantum speedups in solving LP problems. To this end, we first formally define the quantum LP solver.

\begin{definition}[$\epsilon_{\mathrm{LP}}$-approximate quantum LP solver]
\label{defn:lp solver setting}
Consider a linear programming (LP) problem in the form of \eq{primal-LP}: $\max_{\boldsymbol{\xi \in \mathbb{R}^n},\: \boldsymbol{\xi} \geq {0}}\:  \boldsymbol{r}^\top \boldsymbol{\xi} 
,\:\: \mathrm{s.t.}\: \boldsymbol{C}\boldsymbol{\xi} \leq  \boldsymbol{B}$, where all coefficients are bounded within the range $[-1,1]$. Let $m$ denote the number of variables and $d$ the number of constraints. We define OPT as the optimal value of the LP,  $\boldsymbol{\xi}$ as the optimal primal solution, and $\boldsymbol{\eta}$ as the optimal dual solution. Assume that $\sum_i \xi_i \leq R$ and $\sum_i \eta_i \leq r$. 
Given quantum oracles $\mathcal{O}_{\boldsymbol{r}}, \mathcal{O}_{\boldsymbol{C}}$, and $\mathcal{O}_{\boldsymbol{B}}$ that encode the coefficient matrix of the LP problem such that, 
\begin{equation*}
\begin{aligned}
    \mathcal{O}_{\boldsymbol{r}} \ket{i} \ket{0} &= \ket{i} \ket{r_{i}}, \text{ for $i \in [m]$} \\
    \mathcal{O}_{\boldsymbol{C}} \ket{i} \ket{j}\ket{0} &= \ket{i} \ket{j}\ket{C_{i,j}},\text{ for $i \in [d]$, $j \in [m]$},\\
    \mathcal{O}_{\boldsymbol{B}} \ket{i}\ket{0} &= \ket{i} \ket{B_{i}}, \text{ for $i \in [d]$} .
\end{aligned}
\end{equation*}
An quantum LP solver is supposed to compute an $\epsilon_{\mathrm{LP}}$-optimal and $\epsilon_{\mathrm{LP}}$-feasible solution $\boldsymbol{\xi}^* \geq 0 $, satisfying 
\begin{equation*}
    \vert \boldsymbol{r}^\top \boldsymbol{\xi}^* - \mathrm{OPT} \vert \leq \epsilon_{\mathrm{LP}} \quad \mathrm{and} \quad  \boldsymbol{C}\boldsymbol{\xi}^* \leq  \boldsymbol{B} + \epsilon_{\mathrm{LP}}.
\end{equation*}
\end{definition}

Prior work has established the following result for solving LPs by quantum algorithms:

\begin{lemma}[Quantum LP solver \cite{bouland2023quantum,gao2024logarithmic}] 
\label{lem:quantum solver 1}
 Under the setting of \defn{lp solver setting}, there exists a quantum algorithm that solves LP problems with time complexity $\tilde{\mathcal{O}}\left( \frac{\sqrt{d+m}}{\left( \epsilon_{\mathrm{LP}} /Rr \right)^{2.5} }\right)$. 
\end{lemma}

In \algo{quant-bwk2}, we employ the quantum LP solver defined in \lem{quantum solver 1} to solve the corresponding LP problems. \lem{quantum solver 1} are based on
the quantum speedup for zero-sum games via Gibbs sampling, combined with the reduction of general LP problems to zero-sum games, as established in \cite{van2019quantum}. This method also achieves a quadratic speedup with respect to the number of constraints $d$ and variables $m$ compared to classical approximate LP solvers, whose time complexity—based on the classical zero-sum game algorithm in \cite{grigoriadis1995sublinear}—is $\tilde{\mathcal{O}}\left( \frac{d+m}{\left( \epsilon_{\mathrm{LP}} /Rr \right)^2} \right)$. A more detailed discussion on quantum and classical LP solvers is provided in \append{LP solver}.

\begin{algorithm}[!htbp]
    \SetAlgoLined
    \caption{Problem-Dependent Quantum BwK}
    \label{algo:quant-bwk2}
    \small
    
    \Begin{
    \tcp{Phase \uppercase\expandafter{\romannumeral1}: Identification of $\mathcal{I}^*$ and $\mathcal{J}^{\prime}$ }
    Initialize $\mathcal{I}^* = \mathcal{J}^{\prime} = \emptyset$, $r = 1$, $t=0$, $\boldsymbol{B}^{(0)} = \boldsymbol{B}$\; $\theta=\min\left\{\frac{\min\{1,\sigma^2\}\min\{\chi,\delta\}}{12\min\{m^2,d^2\}},\left(2+\frac{1}{b}\right)^{-2}\cdot\frac{\delta}{5}
 \right\}$, $\epsilon=\frac{\min\{1,\sigma\} \min\{\chi,\delta\}b}{5d^{3/2}}$\;
    \While{$\vert \mathcal{I}^* \vert +\vert \mathcal{J}^{\prime} \vert \leq d $}{
        \For{$i=1 \to m$}{
            $N \leftarrow \frac{ \log T}{r} $, $r\leftarrow r/2$ \; 
            Play arm $i$ for $N$ consecutive rounds (by its quantum reward and resource consumption oracles)\;
            Run QMC to get UCB and LCB for the reward and resource consumption of arm $i$ \;
            $t \leftarrow t + N$ \;
        }
        Solve the LCB problem for $\mathrm{OPT}^L_\mathrm{LP}(t)$ in \eq{lcb-lp}.

        \For{$i \not \in \mathcal{I}^{*}$}{
        Solve the UCB problem for $\mathrm{OPT}_i^U(t)$ in \eq{ucb-i}.
    
        \If{$\mathrm{OPT}^L_{\mathrm{LP}}(t) > \mathrm{OPT}^U_i (t)$}{Update $\hat{\mathcal{I}}^*=\hat{\mathcal{I}}^* \cup \{i\} $}
        }

        \For{$j \not \in \mathcal{J}^{\prime}$}{
        Solve the UCB problem for $\mathrm{OPT}_j^U(t)$ in \eq{ucb-j}.
        
        \If{$\mathrm{OPT}^L_{\mathrm{LP}}(t) > \mathrm{OPT}^U_j (t)$}{Update $\hat{\mathcal{J}}^{\prime}=\hat{\mathcal{J}}^{\prime}\cup\{j\}$}
        }
        Update $t=t+1$ and $ \boldsymbol{B}^{(t)}$ \;
    }

    \tcp{Phase \uppercase\expandafter{\romannumeral2}: Exhausting the binding resources}
    \While{$t \leq \tau $ }
    {
    \If{$\Vert \boldsymbol{C}_{\cdot,i}^L(t)-\boldsymbol{C}_{\cdot, i} \Vert_\infty > \theta$ for some $i \in \mathcal{I}^*$ or $\boldsymbol{B}^{(t-1)}/(T-t+1) \not \in [\boldsymbol{b}-\epsilon,\boldsymbol{b}+\epsilon]$}{Solve the following LP
    \begin{equation} 
    \begin{aligned}
    \label{eq:residual-lp}
        \max_{\boldsymbol{\xi}}\quad & \left(\boldsymbol{r}^U(t-1)\right)^\top \boldsymbol{\xi},\\
        \mathrm{s.t.}\quad  & \boldsymbol{C}^L(t-1) \boldsymbol{\xi} \leq \boldsymbol{B}^{(t-1)},\\
        & \xi_i=0,\:i \notin \hat{\mathcal{I}}^*, \\
        & \boldsymbol{\xi} \geq \boldsymbol{0} .
    \end{aligned}
    \end{equation}}
     \Else{Solve the following LP
    \begin{equation} 
    \begin{aligned}
    \label{eq:residual-lp2}
        \max_{\boldsymbol{\xi}}\quad & \left(\boldsymbol{r}^U(t-1)\right)^\top \boldsymbol{\xi},\\
        \mathrm{s.t.}\quad  
        & \boldsymbol{C}^L(t-1) \boldsymbol{\xi} \leq \boldsymbol{B}^{(t-1)},\\
        & \boldsymbol{C}^L_{\mathcal{J}^*,\mathcal{I}^*}(t-1) \boldsymbol{\xi}_{\mathcal{I}^*} \geq \boldsymbol{B}^{(t-1)}_{\mathcal{I}^*},\\
        & \xi_i=0,\:i \notin \hat{\mathcal{I}}^*, \\
        & \boldsymbol{\xi} \geq \boldsymbol{0} .
    \end{aligned}
    \end{equation}}

    Denote its optimal solution as $\tilde{\xi}$ \;
    Normalize $\tilde{\xi}$ into a probability and randomly play an arm according to the probability. (Using classical reward and resource consumption oracles) \;
    Update UCB and LCB of rewards and resources consumption by classical Monte Carlo \;
    Update $\boldsymbol{B}^{(t)}$ and $t=t+1$ \;
    }
    }
\end{algorithm}
\vspace{4mm}

We present the pseudocode of our algorithm in \algo{quant-bwk2}. At time step $t$, we denote by $\boldsymbol{r}^U(t)$ and $\boldsymbol{r}^L(t)$ the UCB and LCB of the reward vector, and by $\boldsymbol{C}^U(t)$ and $\boldsymbol{C}^L(t)$ the UCB and LCB of the resource consumption matrix, respectively. We let $\boldsymbol{B}^{(t)}$ denote the remaining resource budget at time $t$. The algorithm comprises two phases. In Phase \uppercase\expandafter{\romannumeral1}, our objective is to identify all optimal arms $i \in \mathcal{I}^* $ and non-binding constraints $j \in \mathcal{J}^\prime $. In each iteration of the while loop in Phase \uppercase\expandafter{\romannumeral1}, the algorithm plays all arms multiple times and queries the quantum reward and resource consumption oracles. Subsequently, the quantum Monte Carlo method is employed to compute the UCB and LCB for the reward and resource consumption of each arm. To formalize the decision-making process, we define a sequence of LPs that are solved in this phase. First, we introduce $\mathrm{OPT}^L_{\mathrm{LP}}(t)$, which provides a conservative estimate of the achievable reward under the UCB estimates of resource consumption and LCB estimates of reward:
\begin{equation}
\label{eq:lcb-lp}
\begin{aligned}
        \mathrm{OPT}^L_{\mathrm{LP}}(t):=
        \max_{\boldsymbol{\xi}}\quad & {\boldsymbol{r}^L}(t)^\top \boldsymbol{\xi} \\
        \mathrm{s.t.}\quad & {\boldsymbol{C}^U(t)} \boldsymbol{\xi} \leq  \boldsymbol{B}. \\
        & \boldsymbol{\xi} \geq {0}.
\end{aligned}
\end{equation}
Next, for each arm $i$, we formulate an LP to evaluate whether excluding this will degrade the optimal value.
\begin{equation}
\label{eq:ucb-i}
\begin{aligned}
    \mathrm{OPT}_i^U(t):=
    \max_{\boldsymbol{\xi}} \quad & {\boldsymbol{r}^U(t)}^\top \boldsymbol{\xi}, \\
    \mathrm{s.t.} \quad & \boldsymbol{C}^L(t) \boldsymbol{\xi} \leq \boldsymbol{B}, \\
    & \xi_i=0, \boldsymbol{\xi} \geq \boldsymbol{0}.
\end{aligned}
\end{equation}
Finally, to identify non-binding constraints, we formulate the following dual LP:
\begin{equation}
\label{eq:ucb-j}
    \begin{aligned}
    \mathrm{OPT}_j^U(t) :=
    \min_{\boldsymbol{\eta}}\quad & \boldsymbol{B}^\top \boldsymbol{\eta} - B,\\
    \mathrm{s.t.}\quad & {\boldsymbol{C}^L(t)}^\top \boldsymbol{\eta} \geq \boldsymbol{r}^U(t) +\boldsymbol{C}_{j,\cdot}^U(t),\\
    &\boldsymbol{\eta} \geq0 .
\end{aligned}
\end{equation}
As we know in \prop{optimal gap}, there exists a gap in the optimal value between $\mathrm{OPT}_{\mathrm{LP}}$ and $\mathrm{OPT}_i$ (or $\mathrm{OPT}_j$). Provided that the expectations of $\boldsymbol{r}$ and $\boldsymbol{C}$ lie within the confidence interval, it is evident that no matter classical or quantum algorithms are used to estimate the UCB and LCB, the following conditions hold: $\mathrm{OPT}^L_{\mathrm{LP}}(t) \leq \mathrm{OPT}_{\mathrm{LP}}$, $\mathrm{OPT}_i^U(t) \geq \mathrm{OPT}_i$, and $\mathrm{OPT}_j^U (t)\geq \mathrm{OPT}_j$. Consequently, by solving the UCB and LCB LP problems defined in \eq{ucb-i} and \eq{ucb-j} for each arm $i$ and resource $j$, and by comparing these results with \eq{lcb-lp}, we can identify the optimal arms $i \in \mathcal{I}^* $ and non-binding constraints $j \in \mathcal{J}^\prime $ with high probability.

Phase \uppercase\expandafter{\romannumeral2} of \algo{quant-bwk2} leverages the optimal arms $i \in \mathcal{I}^* $ identified in Phase \uppercase\expandafter{\romannumeral1}. In each round of this phase, an adaptive LP approach \eq{residual-lp} is employed to efficiently utilize the remaining resources constrained by binding conditions. When the additional conditions on $\boldsymbol{C}_{\cdot,i}^L(t)$ and $\boldsymbol{B}^{(t-1)}$ are satisfied—namely, when the estimated resource consumption is accurate and the budget is being appropriately utilized—the LP in \eq{residual-lp} becomes equivalent to \eq{residual-lp2}, yielding a more structured allocation that guarantees consumption requirements on the submatrix $\boldsymbol{C}^L_{\mathcal{J}^*,\mathcal{I}^*}(t)$ (see \append{proof-regret-2} for details).
In both formulations, the decision variables $\xi_i$ for arms $i \not \in \mathcal{I}^*$ are fixed at zero, and the available capacity is updated to $\boldsymbol{B}^{(t-1)}$ rather than remaining at its initial value $\boldsymbol{B}$. Once the LP is solved, the optimal solution of this LP is then normalized into a probability distribution, according to which an arm is randomly selected in each round. 
In this phase, to fully utilize the remaining resources, we use the quantum oracle for reward and resource consumption as a classical oracle, meaning a measurement is performed after each oracle call. 
Classical Monte Carlo methods are employed to estimate the UCB and LCB for both rewards and resource consumption. This approach, utilizing quantum oracles in a classical manner, is driven by the necessity for precise control over the probability of playing each optimal arm, ensuring the efficient depletion of the remaining resources. 
Conversely, using quantum Monte Carlo methods in this phase would be less efficient. Each update of the UCB and LCB estimates would necessitate multiple pulls of an arm to refine, as quantum Monte Carlo methods rely on the superposition of quantum states and cannot directly improve these estimates from a single pull. Consequently, solving the LP in each round while pulling an arm multiple times—rather than just once—would reduce the probability of selecting other arms, thereby compromising the efficient exhaustion of the remaining resources.

\subsection{Time complexity analysis}
\label{sec:time complexity analysis}

In \algo{quant-bwk2}, the main time complexity arises from building the LP oracles and solving the LP problems. 
Specifically, in Phase \uppercase\expandafter{\romannumeral1} during each round, all arms are played multiple times, leading to updates in the UCB and LCB for both rewards and resource consumptions. Consequently, the coefficients in the LP formulations—namely Eqs.~\eq{lcb-lp}, \eq{ucb-i}, and \eq{ucb-j}—must be updated accordingly. Moreover, Eq.~\eq{ucb-i} needs to be solved for every $i \notin \mathcal{I}^*$ and Eq.~\eq{ucb-j} for every $j \notin \mathcal{J}^\prime$. 
In Phase \uppercase\expandafter{\romannumeral2}, at each time step we need to solve an LP problem to decide which arm to pull. 
In terms of computational complexity, classical LP solvers exhibit certain limitations. While the simplex method is efficient in many practical scenarios, it has exponential worst-case complexity. The interior point method, which produces high-accuracy solutions, requires $n^\omega$ time, where $\omega < 2.371339$ is the matrix multiplication exponent \cite{alman2025more}. 
This becomes less desirable when the number of arms $m$ and the number of resource constraints $d$ are large.

In fact, we will show that \algo{quant-bwk2} is robust to approximation errors in solving the LP problems in \sec{robustness}. Specifically, we prove that using LP solutions that are only $\epsilon_{\mathrm{LP}}$-feasible and $\epsilon_{\mathrm{LP}}$-optimal does not compromise the algorithm's performance or increase its regret. This allows us to employ approximate LP solvers with significantly lower computational complexity, especially in terms of matrix dimensions.
Classically, such approximate solutions can be obtained with a time complexity of $\tilde{O}((m + d)/\epsilon_{\mathrm{LP}}^2)$~\cite{grigoriadis1995sublinear}. In comparison, the quantum LP solver in \lem{quantum solver 1} achieves a time complexity of $\tilde{O}(\sqrt{m + d}/\epsilon_{\mathrm{LP}}^{2.5})$, 
providing a quadratic speedup with respect to the number of variables $m$ and constraints $d$. We summarize the required solution accuracy for ensuring algorithmic robustness and the corresponding time complexity in \prop{time-lp}.

\begin{proposition}
\label{prop:time-lp}
    Using the quantum LP solver in \lem{quantum solver 1}, the LP problems in \eq{lcb-lp},\eq{ucb-i}, and \eq{ucb-j} can be solved to $\epsilon_{\mathrm{LP}}\cdot T$-optimal and $\epsilon_{\mathrm{LP}}\cdot T$-approximate in quantum time complexity $\tilde{\mathcal{O}}\left(\sqrt{m+d}/\epsilon_{\mathrm{LP}}^{2.5} \right)$. The LPs in \eq{residual-lp} and \eq{residual-lp2} can be solved to $\epsilon_{\mathrm{LP}}\cdot (T-t)/\log^2T $-optimal and $\epsilon_{\mathrm{LP}}\cdot (T-t)/\log^2T $-approximate in quantum time complexity $\tilde{\mathcal{O}}\left(\sqrt{m+d}\log^5 T/\epsilon_{\mathrm{LP}}^{2.5}\right)$. 
\end{proposition}
\prop{time-lp} specifies the required precision for approximately solving the LP problems using quantum solvers. The scaling factors of $T$ and $(T-t)/\log^2(T)$ in the precision requirements are due to the scaling of the respective LP problems. A detailed proof of this proposition is provided in \append{proof of time-lp}.
Furthermore, the parameter $\epsilon_{\mathrm{LP}}$ is a problem-dependent constant with a poly-logarithmic dependence on $T$. In \sec{robustness}, we will establish a formal bound on $\epsilon_{\mathrm{LP}}$. Building upon these results, we then derive the expected time complexity of \algo{quant-bwk2}, as stated in \thm{time-complexity}. 
\begin{theorem}[Time complexity]
\label{thm:time-complexity}
The overall time complexity of \algo{quant-bwk2} with the quantum LP solver in \lem{quantum solver 1} is
\begin{equation}
\tilde{\mathcal{O}} \left( md + \frac{(m + d)^{1.5}}{\epsilon_{\mathrm{LP}}^{2.5}} + \left( \frac{\sqrt{m + d}}{\epsilon_{\mathrm{LP}}^{2.5}} + d \right) T \right).
\end{equation}
\end{theorem}
For comparison, when the original classical algorithm of \cite{li2021symmetry} employs interior‑point methods for its LP subroutines~\cite{jiang2020faster}, its overall time complexity is
\begin{equation}
\tilde{\mathcal{O}}\left( \max\{m, d\}^{2.372} (T + m + d) \right).
\end{equation}
\thm{time-complexity} shows that the quantum LP solver yields a polynomial speedup over classical counterparts with respect to the number of arms $m$ and the number of constraints $d$. (Note that for different parameter ranges of $m$ and $d$, especially when $m$ is large, Ref.~\cite{van2021minimum} gave an LP solver with time complexity $\tilde{O}(md + d^{2.5})$, resulting in a classical algorithm with total time complexity $\tilde{O}((md + d^{2.5})(T+m+d))$. Nevertheless, our quantum algorithms still achieve polynomial speedup in both $m$ and $d$.) 
The proof of this theorem is provided in \append{proof-time}.

\subsection{Regret analysis and robustness of solving LP problems approximately} 
\label{sec:robustness}
The following regret analysis establishes that the algorithm is robust to the aforementioned LP approximation—namely, solving LPs approximately does not degrade performance or increase regret. This robustness is formally established through a detailed analysis, culminating in the regret bound stated in \thm{quant-bwk2}. 

\paragraph{Regret analysis of Phase \uppercase\expandafter{\romannumeral1}}
\
In Phase \uppercase\expandafter{\romannumeral1}, the algorithm identifies the set of optimal arms $\mathcal{I}^*$ and the set of non-binding constraints $\mathcal{J}^\prime$. We first establish a bound on the number of times, denoted by $N$, that each arm is played by the end of this phase. Specifically, the proof consists of two steps. First, it is straightforward that $\hat{\mathcal{I}}^*$ and $\hat{\mathcal{J}}^\prime$ will not include sub-optimal arms $i \in \mathcal{I}^\prime$ and binding constraints $j \in \mathcal{J}^*$, as indicated by the following inequalities:
\begin{equation}
\begin{aligned}
    \mathrm{OPT}^U_i(t) \geq \mathrm{OPT}_i &= \mathrm{OPT}_\mathrm{LP} \geq \mathrm{OPT}^L_\mathrm{LP}(t) \quad \mathrm{for}\ i \in \mathcal{I}^\prime, \\
    \mathrm{OPT}^U_j(t) \geq \mathrm{OPT}_j &= \mathrm{OPT}_\mathrm{LP} \geq \mathrm{OPT}^L_\mathrm{LP}(t) \quad \mathrm{for}\ j \in \mathcal{J}^*.
\end{aligned}
\end{equation}

Next, we need to determine how many times each arm must be played so that, with high probability, the conditions $\mathrm{OPT}^U_i(t) < \mathrm{OPT}^L_\mathrm{LP}(t)$ and $\mathrm{OPT}^U_j (t)< \mathrm{OPT}^L_\mathrm{LP}(t)$ hold for optimal arms $i \in \mathcal{I}^*$ and non-binding constraints $j \in \mathcal{J}^\prime$ with high probability using quantum reward and resource consumption oracles. This can be achieved by employing the quantum Monte Carlo method to obtain the UCB and LCB. Specifically, we will prove the following proposition:

\begin{proposition}
\label{prop:phase 1}
In Phase \uppercase\expandafter{\romannumeral1} of \algo{quant-bwk2}, suppose we obtain an $\epsilon_{\mathrm{LP}}T$-optimal and $\epsilon_{\mathrm{LP}}T$-approximate solution for $\mathrm{OPT}^\mathrm{L}_{\mathrm{LP}}(t)$, $\mathrm{OPT}^\mathrm{U}_{i}(t)$, and $\mathrm{OPT}^\mathrm{U}_{j}(t)$, denoted as $\widetilde{\mathrm{OPT}}^\mathrm{L}_{\mathrm{LP}}(t)$, $\widetilde{\mathrm{OPT}}^\mathrm{U}_{i}(t)$, and $\widetilde{\mathrm{OPT}}^\mathrm{U}_{j}(t)$. To ensure that the following inequalities hold with high probability $1 - \frac{2md}{ T^2}$, each arm will be played no more than
\begin{equation}
    \tilde{\mathcal{O}} \left( \left(2+\frac1b\right)\frac{ \sqrt{d} \log T}{\delta-2\epsilon_{\mathrm{LP}}} \right)
\end{equation}
times:

\begin{align}
     &\mathrm{OPT}^U_i(t) - \mathrm{OPT}_i < \frac{\delta T}{2}- \epsilon_{\mathrm{LP}}T \quad \mathrm{for}\ i \in \mathcal{I}^*, \label{eq:opt-ij-1} \\ 
     &\mathrm{OPT}^U_j(t) - \mathrm{OPT}_j \leq \frac{\delta T}{2}-\epsilon_{\mathrm{LP}}T \quad \mathrm{for}\ j \in \mathcal{J}^\prime, \label{eq:opt-ij-2}\\
     &\mathrm{OPT}^L_{\mathrm{LP}}(t) - \mathrm{OPT}_\mathrm{LP} > -\frac{\delta T}{2}+\epsilon_{\mathrm{LP}}T. \label{eq:opt-ij-3}
 \end{align}

Under these conditions, the approximate solutions satisfy

\begin{equation}
\begin{aligned}
    & \widetilde{\mathrm{OPT}}^L_\mathrm{LP}(t) - \widetilde{\mathrm{OPT}}^U_i(t) \geq \mathrm{OPT}^L_\mathrm{LP}(t)- \mathrm{OPT}^U_i(t) - 2\epsilon_{\mathrm{LP}}T > \mathrm{OPT}_\mathrm{LP}- \mathrm{OPT}_i - \delta T  \geq 0, \\
    & \widetilde{\mathrm{OPT}}^L_\mathrm{LP}(t) - \widetilde{\mathrm{OPT}}^U_j(t) \geq \mathrm{OPT}^L_\mathrm{LP}(t) - \mathrm{OPT}^U_j(t) - 2 \epsilon_{\mathrm{LP}}T >\mathrm{OPT}_\mathrm{LP}- \mathrm{OPT}_j - \delta T  \geq 0 .
\end{aligned}
\end{equation}
This implies that if the LP is approximately solved, the optimal arms $\mathcal{I}^*$ and the non-binding constraints $\mathcal{J}^\prime$ can be correctly identified with high probability.
\end{proposition}

By establishing an upper bound on the total number of times each arm is played during Phase \uppercase\expandafter{\romannumeral1}, we derive the corresponding regret bound stated in \prop{regret 1}. The proofs of \prop{regret 1} and \prop{phase 1} are provided in \append{proof-opt-ij}.
\begin{proposition}[Expected regret in Phase \uppercase\expandafter{\romannumeral1}]
\label{prop:regret 1}
With $\epsilon_{\mathrm{LP}} \leq \frac{\delta}{4}$, the expected regret in Phase \uppercase\expandafter{\romannumeral1} of \algo{quant-bwk2} is upper bounded by 
\begin{equation}
    \tilde{\mathcal{O}} \left( \left(2+\frac1b\right)\frac{ m\sqrt{d} \log T}{b\delta} \right).
\end{equation}
\end{proposition}

\paragraph{Regret analysis of Phase \uppercase\expandafter{\romannumeral2}}
In Phase \uppercase\expandafter{\romannumeral2}, the algorithm focuses on exhausting the binding constraints using an approximate LP solver, which contrasts with classical approaches that assume near exact solvers. Our regret analysis demonstrates the robustness of this method. We show that using an approximate quantum LP solver does not increase the overall regret, provided that the approximation error $\epsilon_{\textrm{LP}}$ is sufficiently small.

\begin{proposition}[Expected regret in Phase \uppercase\expandafter{\romannumeral2}]
\label{prop:regret 2}
With $\epsilon_{\mathrm{LP}} \leq \min\{\frac{\sigma \chi \log^2T}{40\min\{m^{\frac32},d^{\frac32}\}},\frac{2b\theta  \min\{\chi,\delta\}\log^2 T}{405}\}$, the expected regret in Phase \uppercase\expandafter{\romannumeral2} of \algo{quant-bwk2} is upper bounded by 
\begin{equation}
    \mathcal{O} \left( \frac{d^4}{b^2 \min \{\chi^2,\delta^2\} \min\{1,\sigma^2\}} \right).
\end{equation}
\end{proposition}
The core of our proof hinges on the adaptive strategy of the algorithm. We formulate a related LP problem, denoted as \eq{residual-lp2}, which incorporates an additional constraint compared to \eq{residual-lp}. We demonstrate that if the resource consumption matrix is accurately estimated, such that $\Vert \boldsymbol{C}_{\cdot,i}^L(t)-\boldsymbol{C}_{\cdot, i} \Vert_\infty \leq \theta$ for all $i \in \mathcal{I}^*$ and the average remaining resource budget is appropriately bounded, i.e., $\boldsymbol{B}^{(t-1)}/(T-t+1) \in [\boldsymbol{b}-\epsilon,\boldsymbol{b}+\epsilon]$, then the LP problem \eq{residual-lp2} becomes equivalent to \eq{residual-lp}. This equivalence ensures that the approximate solution is sufficiently close to the true optimal solution, making it possible to lower-bound the difference in arm-pulling probabilities between the two policies. Consequently, we show that as long as the LP solutions are accurate within a bounded approximation error $\epsilon_{\textrm{LP}}$, the use of an approximate quantum LP solver does not lead to an increase in regret for Phase \uppercase\expandafter{\romannumeral2}. The detailed proof of \prop{regret 2} is provided in \append{proof-regret-2}.

\paragraph{The expected regret upper bound}
By combining the regret bounds from Phase \uppercase\expandafter{\romannumeral1} in \prop{regret 1} and from Phase \uppercase\expandafter{\romannumeral2} in \prop{regret 2}, we obtain the expected regret upper bound \thm{quant-bwk2} for our quantum algorithm.
\begin{theorem}[Main theorem, problem-dependent setting]
\label{thm:quant-bwk2}
    Under \assum{lp}, the expected regret of \algo{quant-bwk2} is upper bounded by
    \begin{equation}
     \tilde{\mathcal{O}}\left( \left(2+\frac1b\right)\frac{m \sqrt{d} }{b \delta} \log T  + \frac{d^4}{b^2 \min \{\chi^2,\delta^2\} \min\{1,\sigma^2\}} \right),
    \end{equation}
    with $\epsilon_{\mathrm{LP}} \leq \min \{ \frac{\sigma \chi \log^2T}{40\min\{m^{\frac32},d^{\frac32}\}}, \frac{2b\theta  \min\{\chi,\delta\}\log^2 T}{405}, \frac{\delta}{4} \}$.
\end{theorem}
In comparison, the regret bound for the classical counterpart is given by
\begin{equation}
    \mathcal{O} \left( \left( 2+\frac1b \right)^2 \frac{md\log T}{b\delta^2} + \frac{d^4}{b^2 \min \{ \chi^2,\delta^2 \} \min \{1,\sigma^2\}}\right).
\end{equation}
\thm{quant-bwk2} shows that by utilizing quantum oracles, we achieve a quadratic improvement in the parameter $\delta$ and $(2+\frac{1}{b})$ in the first term.

\section{Conclusion}
In this paper, we conduct a systematic investigation of quantum algorithms for the Bandits with Knapsacks problem—a fundamental problem in operations research that combines online linear programming with Markov decision processes, and which also generalizes to reinforcement learning. Our primary contribution is the extension of the classical BwK model to the quantum domain. 
Our first major result is a quantum algorithm for the problem-independent BwK setting. This algorithm achieves a regret bound of  $\mathcal{O}\left(\sqrt{\log (dT)}\right)\left(\mathrm{OPT}_{\mathrm{LP}} \sqrt{\frac{m}{B}}\right) + \mathcal{O}\left(m\log(dT)\log(T)\right)$. By employing QMC for more precise estimations, our approach improves upon classical benchmarks by a factor of $(1+\sqrt{B/\mathrm{OPT}_{\mathrm{LP}}})$. 
Our second main result concerns the problem-dependent BwK setting. In this framework, we utilize both QMC and quantum LP solvers to simultaneously enhance the regret bound and reduce the time complexity of classical algorithms. Specifically, compared to the classical algorithms, our algorithm achieves a quadratic improvement in the regret bound with respect to the problem-dependent parameters. It also exhibits a polynomial speedup in time complexity in terms of the number of arms $m$ and constraints $d$. A key feature of our approach is its resilience to approximation errors. We rigorously demonstrate that employing an approximate quantum LP solver does not degrade the regret guarantee, thereby ensuring the algorithm's robustness.
Our results indicate that quantum algorithms can offer significant advantages in solving complex constrained online learning problems such as BwK, demonstrating clear superiority in both regret performance and computational efficiency.

Our work leaves open questions that necessitate further investigation:
\begin{itemize}
    \item Can the quantum speedups we found for BwK be replicated in other online settings with constraints, such as those with combinatorial or submodular objectives?
    \item Our algorithms' optimality is not yet known. What are the fundamental quantum lower bounds on regret and time complexity for the BwK problem?
\end{itemize}

\section*{Acknowledgments}
Y. Su, Z. Yang, and T. Li were supported by the National Natural Science Foundation of China (Grant Numbers 62372006 and 92365117). P. Huang was supported by the National Natural Science Foundation of China (Grant Number 72302007).

\newcommand{\arxiv}[1]{arXiv:\href{https://arxiv.org/abs/#1}{\ttfamily{#1}}\?}\newcommand{\arXiv}[1]{arXiv:\href{https://arxiv.org/abs/#1}{\ttfamily{#1}}\?}\def\?#1{\if.#1{}\else#1\fi}
\providecommand{\bysame}{\leavevmode\hbox to3em{\hrulefill}\thinspace}

\appendix

\newpage

\section{Quantum and classical algorithms for linear programming}
\label{append:LP solver}

Linear programming (LP) is a core topic in convex optimization~\cite{boyd2004convex}. Among methods for solving LP problems, the simplex method introduced by Dantzig \cite{dantzig1997simplex} solves LP problems by traversing the edges of the feasible region and presents efficiency in practice. However, in the worst case, the simplex method may take an exponential number of iterations with respect to the problem size. In terms of polynomial-time algorithms, Khachiyan~\cite{khachiyan1979polynomial} was the first to prove that LP problems can be solved in polynomial time using the ellipsoid method, thereby establishing that LP belongs to the complexity class \textsf{P}. The same conclusion was also established by Karmarkar~\cite{karmarkar1984new} using the interior-point method. At the moment, the state-of-the-art time complexities of classical LP solvers depend on the specific structure of the problem, such as the relationship between the number of variables $m$ and constraints $d$, or the sparsity of the constraint matrix. When $d=\Omega(m)$, Refs.~\cite{cohen2021solving,jiang2020faster} solve the LP in current matrix multiplication time $\tilde{\mathcal{O}}(m^{\omega})$. Another key result is the $\tilde{\mathcal{O}}(md + d^{2.5})$ complexity from Ref.~\cite{van2021minimum}, which is better when $d$ is small. For sparse LP problems, Ref.~\cite{lee2015efficient} proposed an algorithm with time complexity $\tilde{\mathcal{O}}((\mathrm{nnz}(\boldsymbol{A}) + m^2)\sqrt{m})$, where $\mathrm{nnz}(\boldsymbol{A})$ is the number of nonzero elements in $\boldsymbol{A}$.

While the aforementioned classical LP solvers have a poly-logarithmic dependence on the inverse of the desired precision, approximate LP solvers are also available. Specifically, an LP problem can be reduced to a zero-sum matrix game. Without loss of generality, assume that all coefficients in the LP are within the range $[-1,1]$. 
Denote the optimal value of the LP as OPT, the optimal primal solution as $\boldsymbol{\xi}$, and the optimal dual solution as $\boldsymbol{\eta}$. Assume that $\sum_i \xi_i \leq R$ and $\sum_i \eta_i \leq r$. In this case, it follows that the optimal value, OPT, is bounded by $R$, i.e., $-R \leq \text{OPT} \leq R$. 
The reduction process involves transforming the LP problem into a decision problem: determining whether 
OPT satisfies $\text{OPT} > \alpha$ or $\text{OPT} \leq \alpha + \epsilon$. By performing binary search over the range $[-R,R]$, the LP problem can be solved in $\log(R/\epsilon)$ iterations. 
Furthermore, the decision problem per se can be reduced to computing an $\epsilon$-approximate Nash equilibrium of the corresponding zero-sum game. In particular, \cite[Lemma 12]{van2019quantum} shows that solving the zero-sum games up to additive error $\epsilon^{\prime\prime} = \epsilon/(6R(r+1))$ suffices to correctly conclude the decision problem. Since Grigoriadis and Khachiyan~\cite{grigoriadis1995sublinear} proposed a classical algorithm that finds an $\epsilon$-approximate Nash equilibrium of an $m\times d$ zero-sum game with time complexity $\tilde{\mathcal{O}}((m+d)/\epsilon^2)$, 
it implies a clssical algorithm for finding an $\epsilon$-optimal and $\epsilon$-feasible solution of an LP problem with time complexity $\tilde{\mathcal{O}}((m+d)(R(r+1)/\epsilon)^2)$.
 
On the quantum side, existing approaches for solving LPs are approximate and primarily offer polynomial improvements in complexity with respect to the problem dimensions. The best-known quantum algorithm for computing an $\epsilon$-approximate Nash equilibrium in zero-sum games  utilizes a faster Gibbs sampler and achieves a time complexity of $\tilde{\mathcal{O}}(\sqrt{m+n}/\epsilon^{2.5})$~\cite{bouland2023quantum,gao2024logarithmic}. By leveraging the reduction from LP to zero-sum games, this leads to the development of a quantum LP solver with improved time complexity of $\tilde{\mathcal{O}}(\sqrt{m+d}({R(r+1)}/{\epsilon})^{2.5})$. There also exist other quantum LP solvers with different complexity forms. For instance, Augustino et al.~\cite{augustino2023quantum} proposed a quantum central path algorithm for LPs with complexity $\tilde{\mathcal{O}}(\sqrt{m+d}\cdot\mathrm{nnz}(\boldsymbol{A})R/\epsilon)$. More recently, Augustino et al.~\cite{augustino2025fast} proposed a quantum LP solver based on a quantum version of the mirror descent method with time complexity $\tilde{\mathcal{O}}(d\sqrt{m}({R(r+1)}/{\epsilon})^{2})$.

\section{Proofs for Problem-Independent Quantum BwK}
\label{append:proof in algo1}

In \algo{quant-bwk1} for problem-independent quantum BwK, both quantum and classical estimators are utilized to construct the UCB and LCB for reward and resource consumption. To analyze the algorithm and derive the upper bound of regret, we first present the following propositions, which demonstrate that, with high probability, the expected reward and resource consumption lie within the confidence intervals of the estimator. We begin by deriving the number of times the QMC has been invoked during the execution of the algorithm. Although similar proofs have been presented in \cite{wan2023quantum}, we provide a self-contained proof of \lem{algo1-times} here.

\begin{lemma}
\label{lem:algo1-times}
Let $\delta=\frac{1}{T^2}$ and $C_1$ be the constant in \lem{quant-uni}. The number of times that the quantum Monte Carlo method is invoked is at most 
\begin{equation}
m \log \left( \frac{T}{2mC_1\log T} + 1 \right).
\end{equation}
\end{lemma}
\begin{proof}
Given that there are $T$ total rounds, in Phase \uppercase\expandafter{\romannumeral1} each arm is pulled $2C_1 \log T$ times, and the QMC is invoked $m$ times. In Phase \uppercase\expandafter{\romannumeral2}, suppose each arm invokes the QMC $N_i$ times. Summing the number of times each arm is played across the two phases, we obtain the inequality $\sum_{i=1}^m \left( 2^{N_i + 1} -1  \right)2C_1 \log T \leq T$. 
Next, due to the convexity of the exponent function, $\sum_{i=1}^m  2^{N_i + 1} \geq m 2^{1/m \sum_i N_i + 1}$. 
Therefore, we have
\begin{equation}
    \sum_{i=1}^m (N_i + 1)  \leq m \log \left( \frac{T}{2mC_1\log T} + 1 \right) \leq m \log T,
\end{equation}
which proves the lemma.
\end{proof}
Thus, the QMC method for estimating the reward has been invoked at most $m \log T$ times. With $\delta = \frac{1}{T^2}$ in \lem{quant-uni}, the probability that the expected reward lies within each confidence interval is at least $1- \frac{m \log T}{T^2}$.  Furthermore, the classical estimator for the $d$-dimensional resource consumption has been invoked at most $T$ times. With $\delta=\frac{2d}{T^3}$ in \lem{classic-monto},  the expected resource consumption lies within the respective confidence interval with probability at least $1-\frac{2d}{T^2}$. Consequently, with probability 
$
    1-\frac{m \log T+2d}{T^2}
$, both the expected reward and resource consumption lie within their respective confidence bounds, thereby proving the following proposition.

\begin{proposition}
\label{prop:bounds}
Let $\hat{r}_i(s)$ denote the most recent output from QMC executions up to time $s$, and let $n^{\mathrm{QMC}}_i(s)$ represent the number of quantum oracle queries utilized in the QMC process. For each arm $i \in [m]$ at time $s \in [T]$, we define the following upper and lower confidence bounds for the reward:
\begin{equation}
r^U_i(s) = \mathrm{proj}_{[0,1]}\left(\hat{r}_i(s) + \frac{2C_1 \log T}{n^{\mathrm{QMC}}_i(s)}\right), \quad
r^L_i(s) = \mathrm{proj}_{[0,1]}\left(\hat{r}_i(s) - \frac{2C_1 \log T}{n^{\mathrm{QMC}}_i(s)}\right),
\end{equation}
where $\mathrm{proj}_{[0,1]}$ denotes the projection onto the interval $[0,1]$.
Let $\hat{C}_{j,i}(s)$ denote the empirical mean of resource $j$ observations, and let $n_i(s)$ denote the number of times arm $i$ has been selected by time $s$. We define the following upper and lower bounds for the resource consumption:
\begin{equation}
C^U_{j,i}(s) = \hat{C}_{j,i}(s) + \sqrt{\frac{3\log T}{n_i(s)}}, \quad
C^L_{j,i}(s) = \hat{C}_{j,i}(s) - \sqrt{\frac{3\log T}{n_i(s)}}.
\end{equation}
Then, with probability at least $1-\frac{ m \log T + 2d}{T^2}$, the expected reward and resource consumption are contained within their respective confidence intervals:
\begin{equation}
\begin{gathered}
   r_i \in [r^L_i (s), r^U_i (s)],\: i \in [m],\: s \in [T], \quad
   C_{j,i} \in [C^L_{j,i}(s), C^U_{j,i}(s)],\: i \in [m],\: j \in [d],\: s \in [T]
\end{gathered}
\end{equation}
\end{proposition}

\paragraph{Proof of \thm{main-independent}}
Once we have established the high‐probability confidence intervals in \prop{bounds}, we essentially follow the proof outline of Ref.~\cite{badanidiyuru2018bandits} while incorporating the tighter confidence bounds afforded by QMC to achieve an improved quantum regret bound.

Let the vector $\boldsymbol{z}(s)$ denote the $i_s$-th coordinate vector. Then, the expressions $\boldsymbol{r}^U(s) \boldsymbol{z}(s)$ and $\boldsymbol{C}^L(s) \boldsymbol{z}(s)$ represent the UCB of the reward and the LCB of the resource consumption vector, respectively, for the chosen arm $i(s)$.
Let $\boldsymbol{y}(s)=\boldsymbol{v}(s) / \Vert \boldsymbol{v}(s)\Vert_1$. The following equation satisfies in each stage $s$, where $\boldsymbol{\xi}^*$ is the optimal solution of $\mathrm{OPT}_\mathrm{LP}$ defined in \eq{primal-LP}:
\begin{equation}
\boldsymbol{z}(s) \in \arg\max_{\Vert \boldsymbol{z} \Vert_1=1} \frac{\boldsymbol{r}^U(s)^{\top} \boldsymbol{z}}{\boldsymbol{y}(s)^{\top} \boldsymbol{C}^L(s) \boldsymbol{z}},   
\end{equation}
\begin{equation}
    \frac{\boldsymbol{r}^U(s)^{\top} \boldsymbol{z}(s)}{\boldsymbol{y}(s)^{\top} \boldsymbol{C}^L(s) \boldsymbol{z}(s)} \geq \frac{\boldsymbol{r}^U(s)^{\top} \boldsymbol{\xi}^*}{\boldsymbol{y}(s)^{\top} \boldsymbol{C}^L(s) \boldsymbol{\xi}^*}.
\end{equation}
Summing over $s$, we have
\begin{equation}
\label{eq:left-right}
    \sum_{s=1}^{\tau-1} (\boldsymbol{r}^U(s)^{\top} \boldsymbol{z}(s))(\boldsymbol{y}(s)^{\top} \boldsymbol{C}^L(s) \boldsymbol{\xi}^*) \geq \sum_{s=1}^{\tau-1} (\boldsymbol{r}^U(s)^{\top} \boldsymbol{\xi}^*)(\boldsymbol{y}(s)^{\top} \boldsymbol{C}^L(s) \boldsymbol{z}(s)).
\end{equation}
Then, the right side of the above equation satisfies:  
\begin{equation}
\begin{aligned}
\label{eq:left}
    \sum_{s=1}^{\tau-1} (\boldsymbol{r}^U(s)^{\top} \boldsymbol{\xi}^*)(\boldsymbol{y}(s)^{\top} \boldsymbol{C}^L(s) \boldsymbol{z}(s))& \geq \sum_{s=1}^{\tau-1} (\boldsymbol{r}^{\top} \boldsymbol{\xi}^*)(\boldsymbol{y}(s)^{\top} \boldsymbol{C}^L(s) \boldsymbol{z}(s)) 
    = \mathrm{OPT}_{\mathrm{LP}} \sum_{s=1}^{\tau-1}(\boldsymbol{y}(s)^{\top} \boldsymbol{C}^L(s) \boldsymbol{z}(s)).
\end{aligned}
\end{equation}
Here, the first line comes from $\boldsymbol{r}$ lies in the confidence interval, and the second line is due to $\mathrm{OPT}_\mathrm{LP}=\boldsymbol{r}^\top \boldsymbol{\xi}^*$. 

Next, we introduce the regret bound of multiplicative weight, a folklore result in online learning. 
\begin{proposition}[\cite{freund1997decision}]
\label{prop:online-learning}
    For any fixed $\epsilon$, $\forall \boldsymbol{y} \in [0,1]^d$ satisfying $\Vert \boldsymbol{y} \Vert_1=1$ and any sequence of $\boldsymbol{C}^L(s) \boldsymbol{z}(s)$, we have
    \begin{equation}
    \sum_{s=1}^{\tau-1} (\boldsymbol{y}(s)^{\top} \boldsymbol{C}^L(s) \boldsymbol{z}(s)) \geq 
    (1-\epsilon)\sum_{s=1}^{\tau-1} (\boldsymbol{y}^{\top} \boldsymbol{C}^L(s) \boldsymbol{z}(s)) -  \ln d / \epsilon.
    \end{equation}
\end{proposition}

For the left-hand side of Eq.~\eq{left-right}, we define the vector $\bar{\boldsymbol{y}} = \frac{1}{\mathrm{REW}_\mathrm{UCB}} \sum_s (\boldsymbol{r}^U(s)^\top \boldsymbol{z}(s))\boldsymbol{y}(s)^\top$, where $\mathrm{REW}_\mathrm{UCB} = \sum_s (\boldsymbol{r}^U(s)^\top \boldsymbol{z}(s))$ represents the UCB of total reward in the stages $s$. Consequently, we have $\bar{\boldsymbol{y}} \in [0,1]^d, \Vert \bar{\boldsymbol{y}} \Vert_1 = 1$. Therefore, the following inequality holds:
\begin{align}\label{eq:right}
    \sum_s (\boldsymbol{r}^U(s)^{\top} \boldsymbol{z}(s))(\boldsymbol{y}(s)^{\top} \boldsymbol{C}^L(s) \boldsymbol{\xi}^*) \leq  \sum_s (\boldsymbol{r}^U(s)^{\top} \boldsymbol{z}(s))(\boldsymbol{y}(s)^{\top} \boldsymbol{C}  \boldsymbol{\xi}^*)
    ={\mathrm{REW}_\mathrm{UCB}} \bar{\boldsymbol{y}}^{\top} \boldsymbol{C} \boldsymbol{\xi}^* \leq B\cdot\mathrm{REW}_\mathrm{UCB}.
\end{align}

After successfully deriving both the left and right sides of \eq{left-right}, we proceed to combine them in order to establish the upper bound on regret. Before doing so, we introduce the following proposition, which bounds the difference between the LCB of the cumulative resource consumption and the actual resource consumption. The derivation of this proposition follows the same approach as \cite[Lemma 5.6]{badanidiyuru2018bandits}. Our \prop{bounds} ensure the accuracy of the classical estimator for the resource consumption.

\begin{proposition}[{\cite[Lemma 5.6]{badanidiyuru2018bandits}}] Let $K=\mathcal{O} (\log ( d T))$ and  $\boldsymbol{c}_s$ be the random resource consumption of arm $i(s)$ at time $s$. Then we have the following high probability events.
\label{prop:mean-mu}
    \begin{equation}
        \Vert \sum_{s=1}^{\tau-1}   (\boldsymbol{c}_{s}- \boldsymbol{C}^L(s) \boldsymbol{z}(s)) \Vert_\infty
        \leq
       \mathcal{O}\left( \sqrt{K m B} + K m \log T \right).
    \end{equation}
\end{proposition}

Now, we combine the results of \prop{mean-mu}, \eq{left}, and \eq{right} and apply them to \eq{left-right}. Denote $c_{s}$ as the random resource consumption in stage $s$, and set $\epsilon=\sqrt{\frac{\ln d}{B}}$. We have
\begin{equation}
\begin{aligned}
\label{eq:cost}
    \mathrm{REW}_\mathrm{UCB} 
    &\geq \mathrm{OPT}_{\mathrm{LP}} \left(\frac{(1-\epsilon)}{B}\sum_{s=1}^{\tau-1} (\boldsymbol{y}^{\top} \boldsymbol{C}^L(s) \boldsymbol{z}(s)) - \frac{\ln d}{B \epsilon} \right) \\
    & =\mathrm{OPT}_{\mathrm{LP}} \left( \frac{(1-\epsilon)}{B} \sum_{s=1}^{\tau-1} \boldsymbol{y}^{\top} \boldsymbol{c}_{s}
    - \frac{(1-\epsilon)}{B} \sum_{s=1}^{\tau-1}  \boldsymbol{y}^{\top} (\boldsymbol{c}_{s}-\boldsymbol{C}^L(s) \boldsymbol{z}(s)) 
    - \frac{\ln d}{B \epsilon} \right) \\
    & \geq \mathrm{OPT}_{\mathrm{LP}} \left( \frac{(1-\epsilon)}{B} (B-2m C_1 \log T)
    - \frac{1}{B} \Vert \sum_{s=1}^{\tau-1}  (\boldsymbol{c}_{s}-\boldsymbol{C}^L(s) \boldsymbol{z}(s)) \Vert_\infty
    - \frac{\ln d}{B \epsilon} \right) \\
    & \geq \mathrm{OPT}_{\mathrm{LP}} \left( \frac{(1-\epsilon)}{B} (B-2m C_1 \log T)
    - \frac{\ln d}{B \epsilon} \right) 
    - \mathcal{O}\left(\sqrt{\log (dT)}\right)\left(\mathrm{OPT}_{\mathrm{LP}} \sqrt{\frac{m}{B}}\right)-\mathcal{O}\left(m\log(dT)\log T\right)  \\
    & \geq \mathrm{OPT}_{\mathrm{LP}} \left(
    1- 2\sqrt{\frac{\ln d}{B}} - \frac{2m C_1 \log T}{B}
    \right)
    - \mathcal{O}\left(\sqrt{\log (dT)}\right)\left(\mathrm{OPT}_{\mathrm{LP}} \sqrt{\frac{m}{B}}\right)-\mathcal{O}\left(m\log(dT)\log T\right), 
\end{aligned}
\end{equation}

The third inequality follows from the fact that, upon termination of the algorithm, at least one resource must be exhausted. Let us choose the fixed strategy $\boldsymbol{y}$ to be the unit vector $\boldsymbol{e}_i$ corresponding to the exhausted resource $i$. Then, immediately after the preparation stage, the total consumption of that resource is guaranteed to be at least $B-2mC_1\log T$.

Let $\mathrm{REW}$ denote the expected total reward, and let $N_i$ represent the number of times the QMC method is invoked for arm $i$. Based on \prop{bounds}, the following inequality holds: 
\begin{equation}
\begin{aligned}
\label{eq:rew}
    \mathrm{REW} &\geq \sum_{s=1}^{\tau-1} r_{i(s)} \\
    & = \sum_{s=1}^{\tau-1}( r_{i(s)} - (\boldsymbol{r}^U(s)^{\top} \boldsymbol{z}(s)) +  (\boldsymbol{r}^U(s)^{\top} \boldsymbol{z}(s)) ) \\
    & = \mathrm{REW}_\mathrm{UCB} - \sum_{s=1}^{\tau-1} ( (\boldsymbol{r}^U(s)^{\top} \boldsymbol{z}(s)) - r_{i(s)} )\\ 
    & \geq \mathrm{REW}_\mathrm{UCB} - 
    \sum_{i=1}^{m} \sum_{n=1}^{N_i} \mathrm{rad}_{i} \frac{2C_1}{\mathrm{rad}_{i}/2} \log T\\ 
    & = \mathrm{REW}_\mathrm{UCB} - 4 C_1 \log T
    \sum_{i=1}^{m} N_i \\ 
    & \geq \mathrm{REW}_\mathrm{UCB}
    -4 m C_1 \log T   \log T. 
\end{aligned}
\end{equation}

Combining \eq{cost}, \eq{rew}, the total regret satisfies
\begin{align}
    \mathrm{Regret} & \leq \mathrm{OPT}_{\mathrm{LP}} - \mathrm{REW} \nonumber \\
    & \leq \mathrm{OPT}_{\mathrm{LP}} \left(
    2 \sqrt{\frac{\ln d}{B}} 
    + \frac{2m C_1 \log T}{B} \right) 
    + 4 m C_1 \log^2 T 
    + \mathcal{O}\left(\sqrt{\log (dT)}\right)\left(\mathrm{OPT}_{\mathrm{LP}} \sqrt{\frac{m}{B}}\right)+\mathcal{O}\left(m\log(dT)\log T\right)\nonumber\\
    & = \mathcal{O}\left(\sqrt{\log (dT)}\right)\left(\mathrm{OPT}_{\mathrm{LP}} \sqrt{\frac{m}{B}}\right) + \mathcal{O}\left(m\log(dT)\log T\right),\label{eq:theorem-2-last-line}
\end{align}
where we apply the fact $m\leq B/\log(dT)$ in Eq.~\eq{theorem-2-last-line}. This concludes \thm{main-independent}.

\section{Proofs for Problem-Dependent Quantum BwK}
In this section, we will give proof details for the problem-dependent quantum BwK algorithm (\algo{quant-bwk2}), including its time complexity analysis as well as regret analysis.

\subsection{Proof of \prop{time-lp}}
\label{append:proof of time-lp}

Recall the original LP defined in Eq.~\eq{lcb-lp}:
\begin{equation*}
        \mathrm{OPT}^L_{\mathrm{LP}}(t):=
        \max_{\boldsymbol{\xi}}\  {\boldsymbol{r}^L}(t)^\top \boldsymbol{\xi},\ 
        \mathrm{s.t.}\ {\boldsymbol{C}^U(t)} \boldsymbol{\xi} \leq  \boldsymbol{B},\ 
         \boldsymbol{\xi} \geq {0}, 
\end{equation*}
which needs to be solved to $\epsilon_{\mathrm{LP}} \cdot T$-optimality and $\epsilon_{\mathrm{LP}} \cdot T$-approximation. Since existing quantum approximate LP solvers require all coefficients to lie within $[-1,1]$, we rescale the above LP by a factor of $T$ and solve the following equivalent formulation with accuracy $\epsilon_{\mathrm{LP}}$:
\begin{equation}
        \max_{\boldsymbol{\xi}^{\prime}}\   {\boldsymbol{r}^L}(t)^\top \boldsymbol{\xi}^{\prime},\ 
        \mathrm{s.t.}\  {\boldsymbol{C}^U(t)} \boldsymbol{\xi}^{\prime} \leq  \boldsymbol{B}/T, \ 
         \boldsymbol{\xi}^{\prime} \geq {0}.
\end{equation}
Under this scaling, the LP coefficients satisfy the boundedness condition in $[-1,1]$, and the solution satisfies $\sum_i \xi^{\prime}_i \leq 1$. Let $\boldsymbol{\eta}^\prime$ denote the dual variable; then it holds that $\sum_i \eta^{\prime}_i \leq 1/B$. By existing reduction from general LPs to zero-sum games \cite[Section 4]{van2019quantum}, solving the LP within $\epsilon_{\mathrm{LP}}$ accuracy requires solving the corresponding zero-sum game to additive error $\epsilon^{\prime\prime} = (\epsilon_{\mathrm{LP}}) / (6(\frac{1}{B}+1))$. The associated zero-sum game formulation is:
\begin{equation}
    \min_{\boldsymbol{y}^{\prime\prime} \in \Delta^{n+2}, \lambda \in \mathbb{R}} \lambda, \quad
    \mathrm{s.t.}\ \boldsymbol{A}^{\prime\prime} \boldsymbol{y}^{\prime\prime} \leq \lambda \boldsymbol{e}, 
\end{equation}
where the matrix $\boldsymbol{A}^{\prime\prime}$ is given by:
\begin{equation}
\boldsymbol{A}^{\prime\prime} = 
    \begin{pmatrix}
    \boldsymbol{e}^\top & 1 & -1 \\
    -\boldsymbol{e}^\top & 1 & 1 \\
    -{\boldsymbol{r}^L(t)}^\top & 0 & \alpha  \\
    \boldsymbol{C}^{U}(t) & 0 & -\boldsymbol{B}/T
    \end{pmatrix}.
\end{equation}
The quantum algorithm for solving zero-sum games achieves a time complexity of $\tilde{\mathcal{O}}\left( \sqrt{m+n} \cdot \epsilon^{-2.5} \right)$ \cite{bouland2023quantum,gao2024logarithmic}. Applying this algorithm to the scaled LP yields a total running time of:
\begin{equation}
    \log(1/(\epsilon_{\mathrm{LP}}/T)) \tilde{\mathcal{O}}\left(\sqrt{m+n}\frac{(1\cdot(\frac{1}{B}+1))^{2.5}}{(\epsilon_{\mathrm{LP}})^{2.5}}
    \right) 
    = \tilde{\mathcal{O}}\left(\frac{\sqrt{m+n}}{\epsilon_{\mathrm{LP}}^{2.5}} \right).
\end{equation}

The same rescaling technique with factor $T$ can be applied to the LPs in \eq{ucb-i} and \eq{ucb-j}, and with factor $T - t$ to those in \eq{residual-lp} and \eq{residual-lp2}, yielding the desired quantum time complexity in each case.

\subsection{Proof of \thm{time-complexity}}
\label{append:proof-time}

The quantum time complexity in each round of Phase \uppercase\expandafter{\romannumeral1} is determined by two main components: the cost of QMC and the cost of quantum LP solver. For QMC, if an arm is queried $N$ times, the quantum time complexity for estimating its expected outcome is $\tilde{\mathcal{O}}\left(N + d\log d\right)$. Since each arm is pulled $\tilde{\mathcal{O}}\left( \sqrt{d} \log T \right)$ times with high probability according to \prop{phase 1}, and the number of pulls per arm increases geometrically across rounds, the total sampling complexity is bounded by $\sum_{i=0}^{\mathcal{O}( \log( \sqrt{d}\log T) )} m( 2^i + d\log d) = \mathcal{O}\left( m\sqrt{d}\log T  + m d \log d \log( \sqrt{d}\log T) )\right)$. 
The quantum LP solver in \lem{quantum solver 1} requires oracle access to the coefficients of the LP instance. In each round, the algorithm solves $\mathcal{O}(m + d)$ LPs. To support quantum access to these LPs, the algorithm constructs the required oracles based on the outputs of the QMC procedure, which incurs an additional cost of $\tilde{\mathcal{O}}(md)$ for preparing the relevant QRAM structures. Combining these components, the total quantum time complexity of Phase \uppercase\expandafter{\romannumeral1} is:
\begin{equation}
\begin{aligned}
\label{eq:time quantum 1-1}
    \mathcal{O}( \log( \sqrt{d}\log T) ) \left( \tilde{\mathcal{O}}(md) + \mathcal{O}(m+d) \tilde{\mathcal{O}}\left(\frac{\sqrt{m+d}}{\epsilon_{\mathrm{LP}}^{2.5}}  \right) \right) 
    + \mathcal{O}\left( m\sqrt{d}\log T  + m d \log d \log( \sqrt{d}\log T) )\right).
\end{aligned}
\end{equation}
In Phase \uppercase\expandafter{\romannumeral2}, the algorithm solves an LP problem $\mathcal{O}(T)$ times to exhaust the binding resources. In each round, the algorithm call an quantum oracle and immediately sample it, acting as a classical sample. The estimation will be updated by the classical Monte Carlo methods in \lem{classic-monto} in time $\mathcal{O}(d)$, and updating the QRAM also incurs a cost of $\tilde{\mathcal{O}}(d)$. Consequently, the time complexity in this phase is:
\begin{equation}
\label{eq:time quantum 1-2}
    \mathcal{O}(T) \left(\tilde{\mathcal{O}}(d) + \tilde{\mathcal{O}}\left( \frac{\sqrt{m+d}}{\epsilon_{\mathrm{LP}}^{2.5}}  \right) \right).
\end{equation}
Combining \eq{time quantum 1-1} and \eq{time quantum 1-2}, the overall time complexity of \algo{quant-bwk2} with quantum LP solver in \lem{quantum solver 1} is
\begin{equation}
   \tilde{\mathcal{O}} \left( md + \frac{(m+d)^{1.5}}{\epsilon_{\mathrm{LP}}^{2.5}}+ \left( \frac{\sqrt{m+d}}{\epsilon_{\mathrm{LP}}^{2.5}}  +d\right) T \right).
\end{equation}

\paragraph{Time complexity of the classical algorithm in \cite{li2021symmetry}}
We next analyze the time complexity of the classical two-phase algorithm proposed in \cite{li2021symmetry}, which requires solving LPs to high precision. Unlike the approximate LP solvers considered in our work, this algorithm relies on near-exact solutions. To estimate the cost, we adopt the interior-point method from \cite{jiang2020faster}, which solves LPs in time $\mathcal{O}(n^{2+\min\{\frac{1}{18}, \omega - 2, \frac{1 - \alpha}{2}}\})$, where $\omega$ and $\alpha$ denote the matrix multiplication exponent and its dual, respectively. Under the current state-of-the-art result showing that $\omega<2.371339$~\cite{alman2025more} and $\alpha>0.31389$~\cite{legall2018improved}, this becomes $\mathcal{O}(n^{2.372})$.

In Phase \uppercase\expandafter{\romannumeral1}, the algorithm identifies the optimal arm and the corresponding binding constraints. Each round involves $m$ arm pulls and $\mathcal{O}(m + d)$ LPs. The UCB/LCB updates require $\mathcal{O}(md)$ time per round. Since this phase runs for $\mathcal{O}(\log T)$ rounds, the total complexity is:
\begin{equation}
    \mathcal{O}(\log T)\left( \mathcal{O}(md) +  \mathcal{O}(m+d) \mathcal{O}\left(\max\{m,d\}^{2.372}\right) \right).
\end{equation}
In Phase \uppercase\expandafter{\romannumeral2}, each of the $\mathcal{O}(T)$ rounds involves pulling an arm and solving a single LP.
\begin{equation}
    \mathcal{O}(T)\left( \mathcal{O}(d) + \mathcal{O}\left(\max\{m,d\}^{2.372}\right) \right).
\end{equation}
Combining both phases, the total time complexity of the classical algorithm in \cite{li2021symmetry} is:
\begin{equation}
    \mathcal{O}\left( \max\{m,d\}^{2.372}\left(T + m+d \right) \right).
\end{equation}

\subsection{Proofs of \prop{phase 1} and \prop{regret 1}}
\label{append:proof-opt-ij}
In each pull of arm, the reward is estimated using the univariate quantum Monte Carlo method in \lem{quant-uni}, while the estimation of the $d$-dimension resource consumption vector of each arm uses the multivariate quantum Monte Carlo method in \lem{quant-multi}. 
In the following discussion, we consider the event that all expectations for the reward and resource consumption are bounded by the corresponding UCBs and LCBs produced by the quantum Monte Carlo method. 
In Phase \uppercase\expandafter{\romannumeral1}, suppose each arm is pulled $N(t) \equiv C_2 n(t) \sqrt{\log n(t)}$ times. As a result, both the univariate and multivariate QMC methods are invoked $2m\log N(t)$ times in total. Given that the failure probability of a single QMC invocation is $\delta_{\mathrm{QMC}}$, this setup ensures that the true expectations fall within the estimated confidence bounds with probability at least $1 - 2m \log N(t) \cdot \delta_{\mathrm{QMC}}$.

We then show that if
\begin{equation}
\label{eq:n(t)}
    n(t) \geq \left( 2 + \frac{1}{b} \right)\frac{C_1 \sqrt{d} \log \left( d/\delta_{\mathrm{QMC}}\right) }{\delta/2 - \epsilon_{\mathrm{LP}}} ,
\end{equation}
then Eqs.~\eq{opt-ij-1}, \eq{opt-ij-2}, and \eq{opt-ij-3} hold.

\paragraph{Proof of \eq{opt-ij-1}}
We introduce the following LP to relate $\mathrm{OPT}^U_i(t)$ with $\mathrm{OPT}_i$
\begin{equation}
\begin{aligned}
\label{eq:optu-opti}
    \max_{\boldsymbol{\xi}} \quad & \boldsymbol{r}^\top \boldsymbol{\xi} , \\
    \mathrm{s.t.}\quad & \boldsymbol{C}\boldsymbol{\xi} \leq (1+ \frac{\sqrt{d}\log \left( d/\delta_{\mathrm{QMC}}\right)}{bn(t)} ) \boldsymbol{B}, \\
    & \xi_i = 0,\: \boldsymbol{\xi} \geq 0 . 
\end{aligned}
\end{equation}
It is easy to see that the optimal solution of \eq{optu-opti} is 
$\left( 1 + \frac{\sqrt{d}\log \left( d/\delta_{\mathrm{QMC}}\right)}{bn(t)} \right)\mathrm{OPT}_i $.

Let $\boldsymbol{\xi}^U(t)$ be the optimal solution of $\mathrm{OPT}_i^U(t)$. Then we will prove that $\boldsymbol{\xi}^U(t)$ is a feasible solution of \eq{optu-opti} and build the relation between $\mathrm{OPT}_i^U(t)$ and $\mathrm{OPT}_i$.

First, we prove the feasibility
\begin{equation}
\begin{aligned}
(\boldsymbol{C} \boldsymbol{\xi}^U(t))_j 
& = (\boldsymbol{C}^L(t) \boldsymbol{\xi}^U(t))_j+((\boldsymbol{C}-\boldsymbol{C}^L(t)) \boldsymbol{\xi}^U (t))_j \\
& \leq B + \Vert (\boldsymbol{C}-\boldsymbol{C}^{L}(t))_{j} \Vert_{\infty}  \Vert \boldsymbol{\xi}^{U}(t) \Vert_{1} \\
& \leq B+  \frac{\sqrt{d}\log \left( d/\delta_{\mathrm{QMC}}\right)}{n(t)} T,\\
& = B\left(1+ \frac{\sqrt{d}\log \left( d/\delta_{\mathrm{QMC}}\right)}{bn(t)}\right).
\end{aligned}
\end{equation}

Then we have
\begin{equation}
\begin{aligned}
\mathrm{OPT}_i^U(t)
& = {\boldsymbol{r}^U}(t)^\top \boldsymbol{\xi}^U(t) \\
& \leq \boldsymbol{r}^\top \boldsymbol{\xi}^U(t) + \frac{C_1 \log \left(1/\delta_{\mathrm{QMC}}\right)}{N(t)} T \\
& \leq \left( 1 + \frac{\sqrt{d}\log \left( d/\delta_{\mathrm{QMC}}\right)}{bn(t)} \right)\mathrm{OPT}_i + \frac{ C_1 \log \left(1/\delta_{\mathrm{QMC}}\right)}{N(t)} T \\
& < \mathrm{OPT}_i + \left( 1 + \frac{1}{b} \right)\frac{ C_1 \sqrt{d} \log \left( d/\delta_{\mathrm{QMC}}\right) }{n(t)} T \\
& \leq \mathrm{OPT}_i + (\frac{\delta}{2} - \epsilon_{\mathrm{LP}})T.
\end{aligned}
\end{equation}
Hence we have established \eq{opt-ij-1}.

\paragraph{Proof of \eq{opt-ij-2}}
Let $\boldsymbol{\eta}$ be the optimal solution to $\mathrm{OPT}_j$, we will prove that
\begin{align}
\boldsymbol{\eta}^\prime(t) = \boldsymbol{\eta} +   \left( \frac{C_1 \log \left( 1/\delta_{\mathrm{QMC}}\right)}{bN(t)} + \left(2+\frac1b \right)\frac{\sqrt{d}\log \left( d/\delta_{\mathrm{QMC}}\right)}{bn(t)} \right)(1,0,\ldots,0)^\top
\end{align}
is the feasible solution to $\mathrm{OPT}^U_j$. We have
\begin{equation}
\begin{aligned}
{\boldsymbol{C}^L}^\top \boldsymbol{\eta}^{\prime}(t)
& ={\boldsymbol{C}^L(t)}^\top \boldsymbol{\eta} +  \frac{C_1 \log \left( 1/\delta_{\mathrm{QMC}}\right) }{bN(t)} {\boldsymbol{C}^L(t)}^\top (1,0,\ldots,0)^\top + \left(2+\frac1b \right) \frac{\sqrt{d}\log \left( d/\delta_{\mathrm{QMC}}\right)}{bn(t)} {\boldsymbol{C}^L(t)}^\top (1,0,\ldots,0)^\top \\
& = \boldsymbol{C}^\top \boldsymbol{\eta}+(\boldsymbol{C}^L(t)-\boldsymbol{C})^\top \boldsymbol{\eta} +  \frac{C_1 \log \left( 1/\delta_{\mathrm{QMC}}\right) }{N(t)} \boldsymbol{1}  + \left(2+\frac1b \right) \frac{\sqrt{d}\log \left( d/\delta_{\mathrm{QMC}}\right)}{n(t)} \boldsymbol{1} \\
& \geq \boldsymbol{r} + \boldsymbol{C}_{j,\cdot}^{\top} - \max_{j\in[d]} \Vert \boldsymbol{C}_{j,\cdot}^{L}(t) - \boldsymbol{C}_{j,\cdot} \Vert_{\infty}(\boldsymbol{1}^{\top} \boldsymbol{\eta}) \boldsymbol{1} + \frac{C_1 \log \left( 1/\delta_{\mathrm{QMC}}\right) }{N(t)} \boldsymbol{1}  + \left(2+\frac1b \right) \frac{\sqrt{d}\log \left( d/\delta_{\mathrm{QMC}}\right)}{n(t)} \boldsymbol{1}  \\
& \geq \boldsymbol{r} + \boldsymbol{C}_{j,\cdot}^{\top} -  \frac{\sqrt{d}\log \left( d/\delta_{\mathrm{QMC}}\right)}{n(t)} \left(1+\frac1b\right) \boldsymbol{1}
+ \frac{C_1 \log \left( 1/\delta_{\mathrm{QMC}}\right) }{N(t)} \boldsymbol{1}  + \left(2+\frac1b \right) \frac{\sqrt{d}\log \left( d/\delta_{\mathrm{QMC}}\right)}{n(t)} \boldsymbol{1} \\
& = \boldsymbol{r} + \boldsymbol{C}_{j,\cdot}^{\top} 
+  \frac{C_1 \log \left( 1/\delta_{\mathrm{QMC}}\right) }{N(t)} \boldsymbol{1}  + \frac{\sqrt{d}\log \left( d/\delta_{\mathrm{QMC}}\right)}{n(t)} \boldsymbol{1} \\
& \geq \boldsymbol{r}^{U}(t) + (\boldsymbol{C}_{j,\cdot}^{U}(t))^{\top}.
\end{aligned}
\end{equation}
The forth line is due to $\boldsymbol{B}^\top \boldsymbol{\eta} - B \leq T$ and $\boldsymbol{1}^\top \boldsymbol{\eta} \leq 1 + \frac{1}{b}$. As a result, the following inequalities hold:
\begin{equation}
\begin{aligned}
\mathrm{OPT}_j^U(t) 
& \leq \boldsymbol{B}^{\top} \boldsymbol{\eta}^{\prime}(t) - B\\
&\leq \boldsymbol{B}^\top \boldsymbol{\eta} - B + \left( \frac{C_1 \log \left( 1/\delta_{\mathrm{QMC}}\right)}{N(t)} + \left(2+\frac1b \right)\frac{\sqrt{d}\log \left( d/\delta_{\mathrm{QMC}}\right)}{n(t)} \right) T \\
& \leq \boldsymbol{B}^\top \boldsymbol{\eta} - B  +  \left(2+\frac1b \right)\frac{C_1\sqrt{d}\log \left( d/\delta_{\mathrm{QMC}}\right)}{n(t)}  T \\
& \leq \mathrm{OPT}_j + \left(\frac{\delta}{2}-\epsilon_{\mathrm{LP}}\right)T.
\end{aligned}
\end{equation}
Hence we have established \eq{opt-ij-2}.

\paragraph{Proof of \eq{opt-ij-3}}
Recall that $\boldsymbol{\xi}^*$ is the optimal solution of $\mathrm{OPT}_\mathrm{LP}$, we will prove that $\left(1-\frac{\sqrt{d}\log \left( d/\delta_{\mathrm{QMC}}\right)}{bn(t)} \right)\boldsymbol{\xi}^*$ is the feasible solution to $\mathrm{OPT}^L_\mathrm{LP}(t)$.
\begin{equation}
\begin{aligned}
B - \left(1-\frac{\sqrt{d}\log \left( d/\delta_{\mathrm{QMC}}\right)}{bn(t)} \right) (\boldsymbol{C}^{U}(t) \boldsymbol{\xi}^{*})_{j}
& = B-(\boldsymbol{C}^U(t) \boldsymbol{\xi}^*)_j + \frac{\sqrt{d}\log \left( d/\delta_{\mathrm{QMC}}\right)}{n(t)}(\boldsymbol{C}^U(t) \boldsymbol{\xi}^*)_j \\
& \geq B-(\boldsymbol{C} \boldsymbol{\xi}^*)_j - ((\boldsymbol{C}^U(t)-\boldsymbol{C}) \boldsymbol{\xi}^*)_j + \frac{\sqrt{d}\log \left( d/\delta_{\mathrm{QMC}}\right)}{bn(t)}(\boldsymbol{C} \boldsymbol{\xi}^*)_j \\
&\geq B- (\boldsymbol{C}\boldsymbol{\xi}^*)_j - \Vert (\boldsymbol{C}^U(t)-\boldsymbol{C})_{j,\cdot} \Vert_\infty \Vert \boldsymbol{\xi}^*\Vert_1 + \frac{\sqrt{d}\log \left( d/\delta_{\mathrm{QMC}}\right)}{bn(t)}(\boldsymbol{C}\boldsymbol{\xi}^*)_j \\
&\geq B- (\boldsymbol{C}\boldsymbol{\xi}^*)_j
- \frac{\sqrt{d}\log \left( d/\delta_{\mathrm{QMC}}\right)}{bn(t)} B
+ \frac{\sqrt{d}\log \left( d/\delta_{\mathrm{QMC}}\right)}{bn(t)}(\boldsymbol{C}\boldsymbol{\xi}^*)_j \\
& = \left(1-\frac{\sqrt{d}\log \left( d/\delta_{\mathrm{QMC}}\right)}{bn(t)}\right) (B- (\boldsymbol{C}\boldsymbol{\xi}^*)_j) \\
& \geq 0.
\end{aligned}
\end{equation}

Then we have the following inequality
\begin{equation}
\begin{aligned}
\mathrm{OPT}_{\mathrm{LP}}^{L}(t)
& \geq \left(1-\frac{\sqrt{d}\log \left( d/\delta_{\mathrm{QMC}}\right)}{bn(t)} \right)(\boldsymbol{r}^L(t))^\top \boldsymbol{\xi}^* \\
&=\boldsymbol{r}^\top \boldsymbol{\xi}^* 
+ \left(\boldsymbol{r}^L(t)-\boldsymbol{r}\right)^\top \boldsymbol{\xi}^*
-\frac{\sqrt{d}\log \left( d/\delta_{\mathrm{QMC}}\right)}{bn(t)} {\boldsymbol{r}^L(t)}^\top \boldsymbol{\xi}^* \\
& \geq \mathrm{OPT}_{\mathrm{LP}} 
-  \Vert \boldsymbol{r}^L(t)-\boldsymbol{r} \Vert_\infty \Vert \boldsymbol{\xi}^* \Vert_1
-\frac{\sqrt{d}\log \left( d/\delta_{\mathrm{QMC}}\right)}{bn(t)} {\boldsymbol{r}}^\top \boldsymbol{\xi}^* \\
& \geq \mathrm{OPT}_{{\mathrm{LP}}}
- \frac{C_1 \log \left( 1/\delta_{\mathrm{QMC}}\right)}{N(t)} T - \frac{\sqrt{d}\log \left( d/\delta_{\mathrm{QMC}}\right)}{bn(t)} T \\
& > \mathrm{OPT}_{{\mathrm{LP}}}
- \left( 1 + \frac{1}{b} \right)\frac{ C_1 \sqrt{d} \log \left( d/\delta_{\mathrm{QMC}}\right) }{n(t)} T \\
& > \mathrm{OPT}_{{\mathrm{LP}}} - \left(\frac{\delta}{2} - \epsilon_{\mathrm{LP}}\right)T.
\end{aligned}
\end{equation}
Hence we have established \eq{opt-ij-3}.
\\\\
In conclusion, let $\delta_{\mathrm{QMC}} = \frac{d}{T^3}$ in \eq{n(t)}. Then, with probability at least $1-\frac{2md}{T^2}$, if $n(t) \geq \left( 2 + \frac{1}{b} \right)\frac{3C_1 \sqrt{d} \log T  }{\delta/2 - \epsilon_{\mathrm{LP}}}$, 
then Eqs.~\eq{opt-ij-1}, \eq{opt-ij-2}, and \eq{opt-ij-3} are all satisfied. This completes the proof of \prop{phase 1}. The result guarantees that, under the assumption $\epsilon_{\mathrm{LP}} \leq \frac{\delta}{4}$, each arm is pulled at most $\tilde{\mathcal{O}} \left( \left(2+\frac{1}{b} \right)\frac{ \sqrt{d} \log T}{\delta} \right)$ times during Phase \uppercase\expandafter{\romannumeral1} of \algo{quant-bwk2}.
To upper bound the regret incurred in Phase \uppercase\expandafter{\romannumeral1}, we introduce the following general regret bound for BwK algorithms:
\begin{proposition}[{\cite[Proposition 1]{li2021symmetry}}]
\label{prop:regret-bound}
    The following inequality holds for any BwK algorithm:
    \begin{equation}
    \mathrm{Regret} \leq
    \sum_{i \in \mathcal{I^\prime}} n_i(t) \Delta_i + 
    \mathbb{E} \left[\boldsymbol{B}^{(\tau)}\right]^\top \boldsymbol{\eta}^*,
    \end{equation}
    where $\Delta_i = \boldsymbol{C}_{\cdot,i} \boldsymbol{\eta}^* - r_i$ for $i \in [m]$.
\end{proposition}
This bound indicates that the regret consists of two parts: the cumulative contribution from pulling suboptimal arms, and the expected value of leftover resources weighted by the optimal dual solution. In our two-phase quantum algorithm, suboptimal arms are pulled only in Phase \uppercase\expandafter{\romannumeral1}. Phase \uppercase\expandafter{\romannumeral2} selects only optimal arms and thus contributes regret solely through the second term. Observe that $\Delta_i = \boldsymbol{C}_{\cdot,i}^\top \boldsymbol{\eta}^* - r_i \leq \boldsymbol{1}^\top \boldsymbol{\eta}^*$, and by duality, $\boldsymbol{B}^\top \boldsymbol{\eta}^* \leq T$, implying $\boldsymbol{1}^\top \boldsymbol{\eta}^* \leq \frac{1}{b}$. Therefore, we obtain the bound
\begin{equation}
    \sum_{i \in \mathcal{I^\prime}} n_i(t) \Delta_i = \tilde{\mathcal{O}} \left( \left(2+\frac1b\right)\frac{ m\sqrt{d} \log T}{b\delta} \right),
\end{equation}
which completes the proof of \prop{regret 1}.

\subsection{Robustness Analysis of Phase \uppercase\expandafter{\romannumeral2} of \algo{quant-bwk2} under Approximate LP Solutions}
\label{append:proof-regret-2}

When addressing linear programming problems whose time complexity depends on $1/\epsilon$ rather than $\log(1/\epsilon)$, it is critical to assess how approximations might affect the algorithm's robustness. In this work, we present a regret analysis for Phase \uppercase\expandafter{\romannumeral2} of \algo{quant-bwk2} under approximate LP solutions. Our analysis demonstrates that even when the LP is solved to $\epsilon_{\mathrm{LP}}(T-t)/\log^2 T $-optimal and $\epsilon_{\mathrm{LP}}(T-t)/\log^2 T$-feasible, the regret in Phase \uppercase\expandafter{\romannumeral2} remains a problem-dependent constant, provided that $\epsilon_{\mathrm{LP}}$ is sufficiently small. Such robustness of LP may be of independent interest.
Our analysis is built upon the framework of the pure classical problem-dependent bandit with knapsacks problem as presented in \cite{li2021symmetry}. In particular, with $\theta=\min\left\{\frac{\min\{1,\sigma^2 \}\min\{\chi,\delta\}}{12\min\{m^2,d^2\}},\left(2+\frac{1}{b}\right)^{-2}\cdot\frac{\delta}{5} \right\}$ and $\epsilon=\frac{\min\{1,\sigma\} \min\{\chi,\delta\}b}{5d^{3/2}}$, Phase \uppercase\expandafter{\romannumeral2} of \algo{quant-bwk2} can be decomposed into three stages:
\begin{enumerate}
    \item In initial stage, the optimal arms have not yet been played for a sufficient number of rounds, so that the condition $\Vert \boldsymbol{C}_{\cdot,i}^L(t)-\boldsymbol{C}_{\cdot,i} \Vert_\infty > \theta$ holds. Intuitively, the duration of this stage is of order $\mathcal{O}(\log T)$.
    \item After initial stage, each optimal arm has been sampled sufficiently often, which facilitates an accurate estimation of the resource consumption matrix $\boldsymbol{C}$. This improved estimation ensures that the condition $\Vert \boldsymbol{C}_{\cdot,i}^L(t)-\boldsymbol{C}_{\cdot,i} \Vert_\infty \leq\theta$ continues to hold, thereby guaranteeing the stability of the smallest singular value of $\boldsymbol{C}^L(t)$. Concurrently, we impose that in each round the average remaining resource of binding constraints is maintained within the interval $[\boldsymbol{b}-\epsilon, \boldsymbol{b} + \epsilon]$. Specifically, if $\boldsymbol{B}_{\mathcal{I}^*}^{(t-1)}/(T-t+1) \in [\boldsymbol{b}-\epsilon, \boldsymbol{b} + \epsilon]$, then it follows with high probability that in the subsequent round $t$, the condition  $\boldsymbol{B}_{\mathcal{I}^*}^{(t)}/(T-t) \in [\boldsymbol{b}-\epsilon, \boldsymbol{b} + \epsilon]$ will be satisfied. Intuitively, this stage is expected to extend over $\mathcal{O}(T)$ rounds.
    \item The final stage occurs when the average remaining resource deviates from the target interval, i.e., when $\frac{\boldsymbol{B}_{\mathcal{I}^*}^{(t)}}{T-t}\not \in [\boldsymbol{b}-\epsilon, \boldsymbol{b} + \epsilon]$. This stage occurs with a small probability and is expected to last only $\mathcal{O}(1)$ rounds.
\end{enumerate}
Since all optimal arms have been identified in Phase \uppercase\expandafter{\romannumeral1}, the regret incurred during Phase \uppercase\expandafter{\romannumeral2} arises from the residual resources associated with the binding constraints, as characterized by \prop{regret-bound} in \append{proof-opt-ij}. In the following analysis, we examine the remaining resources of the binding constraints after stage 2, which provides an upper bound on the final residual resources after the completion of all three stages. Our analysis is conducted under the \assum{warm start assum}, which is also adopted in \cite{li2021symmetry}.
\begin{assumption}
\label{assum:warm start assum}
    We make the following assumption on the following analysis:
    \begin{enumerate}
        \item $\boldsymbol{r}^{U}(t) \geq \boldsymbol{r}$, $\boldsymbol{C}^{L}(t) \leq \boldsymbol{C}$ hold element-wise for all $t\in [T]$.
        \item 
        all arms are optimal and all the constraints are binding, i.e., $\mathcal{I}^* = [m],\text{ }\mathcal{J}^* = [d]$. 
    \end{enumerate}
\end{assumption}

Based on \lem{classic-monto}, the first part of \assum{warm start assum} is satisfied with high probability. Note that in this phase, our \algo{quant-bwk2} employs the quantum oracle as a classical oracle, and uses the classical Monte Carlo method to estimate the UCB and LCB of the reward and resource consumption. Let $\hat{\boldsymbol{r}}$ and $\hat{\boldsymbol{C}}_{\cdot,i}$ denote the classical estimates of the reward and consumption of arm $i$. If this arm is pulled $N$ times, we have  $\vert \hat{r}_i -r_i \vert \leq \sqrt{ \frac{2\log T}{N} }$ and $\Vert \hat{\boldsymbol{C}}_{\cdot,i} - \boldsymbol{C}_{\cdot,i} \Vert_\infty \leq  \sqrt{ \frac{2\log T}{N} }$ with probability at least $1 - \frac{4md}{T^2}$ by \cite[Lemma 2]{li2021symmetry}. The second part of \assum{warm start assum} is motivated by the intuition that the non-binding constraints will not be exhausted until the binding constraints are depleted. The relaxation of the second part follows directly from \cite[Appendix C4]{li2021symmetry}.

\subsubsection{Stage 1 of Phase \uppercase\expandafter{\romannumeral2}}
We begin by formulating the following three linear programming problems only for arms $i \in \mathcal{I}^*$. 
\begin{equation}
\begin{aligned}
\label{eq:stage1:eq_i}
\tilde{\text{OPT}}_i^{U}(t)=\max_{\boldsymbol{\xi}} \ \ & \left(\boldsymbol{\mu}^U(t-1) \right)^\top \boldsymbol{\xi},\\
\text{s.t.}\ \ &  \boldsymbol{C}^L(t-1) \boldsymbol{\xi} \le \boldsymbol{B}^{(t-1)},    \\
& \xi_i=0,\ \boldsymbol{\xi}\geq \boldsymbol{0},
\end{aligned}
\end{equation}
\begin{equation}
\begin{aligned}
\label{eq:stage1:eq_i_relaxed}
\tilde{\text{OPT}}_i^{U,\text{Relaxed}}(t)=\max_{\boldsymbol{\xi}} \ \ & \left(\boldsymbol{\mu}^U(t-1) \right)^\top \boldsymbol{\xi},\\
\text{s.t.}\ \ &  \boldsymbol{C}^L(t-1) \boldsymbol{\xi} \le \boldsymbol{B}^{(t-1)} + \epsilon_1,    \\
& \xi_i=0,\ \boldsymbol{\xi}\geq \boldsymbol{0},
\end{aligned}
\end{equation}
\begin{equation}
\begin{aligned}
\label{eq:stage1:eq_lp}
    \tilde{\text{OPT}}^{U}_{\text{LP}}(t)=\max_{\boldsymbol{\xi}} \ \ & \left(\boldsymbol{\mu}^U(t-1) \right)^\top \boldsymbol{\xi},\\
        \text{s.t.}\ \ &  \boldsymbol{C}^L(t-1) \boldsymbol{\xi} \le \boldsymbol{B}^{(t-1)},    \\
        & \boldsymbol{\xi}\geq \boldsymbol{0}.
\end{aligned}
\end{equation}

\begin{lemma}[{\cite[Appendix C3]{li2021symmetry}}] There is a lower bound on the gap between $\tilde{\mathrm{OPT}}^{U}_{\mathrm{LP}}(t)$ and $\tilde{\mathrm{OPT}}_i^{U}(t)$ for $t\leq\frac{b\min\{\chi\theta,\delta\theta\}}{90d}T$: 
\begin{equation*} \tilde{\mathrm{OPT}}^{U}_{\mathrm{LP}}(t) - \tilde{\mathrm{OPT}}_i^{U}(t) \geq \frac{\min\{\chi\theta,\delta\theta\}}{45}(T-t). \end{equation*} \end{lemma}
Now, let us define $\epsilon_1 := \epsilon_{\mathrm{LP}}(T-t)/\log^2 T$. We consider the scenario where the linear program in \eq{stage1:eq_lp} is solved up to an $\epsilon_1$-optimal and $\epsilon_1$-approximate solution. Denote this solution by $\tilde{\boldsymbol{\xi}}$, which satisfies: \begin{equation} \begin{aligned} \tilde{\mathrm{OPT}}^{U}_{\text{LP}}(t) - \epsilon_1 \leq\ &\left(\boldsymbol{\mu}^U(t-1)\right)^\top \tilde{\boldsymbol{\xi}} \leq \tilde{\mathrm{OPT}}^{U}_{\text{LP}}(t) + \epsilon_1, \: \boldsymbol{C}^L(t-1) \tilde{\boldsymbol{\xi}} \leq \boldsymbol{B}^{(t-1)} + \epsilon_1, \: \tilde{\boldsymbol{\xi}} \geq \boldsymbol{0}. \end{aligned} \end{equation}
Our goal is to derive an element-wise lower bound for the approximate solution $\tilde{\boldsymbol{\xi}}$, in order to demonstrate that even with approximate LP solutions, stage 1 can still be completed in $\mathcal{O}(\log T)$ rounds. 
Observe that $\tilde{\boldsymbol{\xi}}$ with $\tilde{\boldsymbol{\xi}}_i = 0$ is feasible for the relaxed version of the LP in \eq{stage1:eq_i_relaxed}. Hence, we have the inequality: 
\begin{equation}
    \left(\boldsymbol{\mu}^U(t-1) \right)^\top \tilde{\boldsymbol{\xi}} = \left(\boldsymbol{\mu}^U(t-1) \right)^\top_{\neq i} \tilde{\boldsymbol{\xi}}_{\neq i} + \left(\boldsymbol{\mu}^U(t-1) \right)^\top_i \tilde{\boldsymbol{\xi}}_i \leq \tilde{\text{OPT}}_i^{U,\text{Relaxed}}(t) + \tilde{\boldsymbol{\xi}}_i.
\end{equation}
Let $\boldsymbol{\xi}$ be an optimal solution to \eq{stage1:eq_i_relaxed}. Then, the scaled vector $\boldsymbol{\xi}/(1+\epsilon_1/\max\{\boldsymbol{B}^{(t-1)}\})$ is a feasible solution for \eq{stage1:eq_i}, implying the bound:
\begin{equation}
    \frac{\tilde{\text{OPT}}_i^{U,\text{Relaxed}}(t)}{1+\epsilon_1/\max\{\boldsymbol{B}^{(t-1)}\}} \leq \tilde{\text{OPT}}_i^{U}(t).
\end{equation}
Combining the inequalities above, we obtain: 
\begin{equation}
\begin{aligned}
    \tilde{\boldsymbol{\xi}}_i &\geq \tilde{\text{OPT}}^{U}_{\text{LP}}(t)- \epsilon_1 - 
     \tilde{\text{OPT}}_i^{U}(t) \left( 1+\epsilon_1/\max\{\boldsymbol{B}^{(t-1)}\} \right)  \\
     & \geq \tilde{\text{OPT}}^{U}_{\text{LP}}(t)-  
     \tilde{\text{OPT}}_i^{U}(t) -\epsilon_1  \left( 1+(T-t)/\max\{\boldsymbol{B}^{(t-1)} \}\right) \\
     & \geq \frac{\min\{\chi\theta,\delta\theta\}}{45}(T-t)  - \frac{\epsilon_{\mathrm{LP}}(T-t)}{\log^2 T}\left( 1+\frac{T-t}{(T-t)\frac{4}{5}b} \right) \\
     & \geq \frac{\min\{\chi\theta,\delta\theta\}}{45}(T-t)  - \frac{9\epsilon_{\mathrm{LP}}(T-t)}{4b\log^2 T} \\
     & \geq \frac{\min\{\chi\theta,\delta\theta\}}{90}(T-t),
\end{aligned}
\end{equation}
where the last line follows from the assumption that $\epsilon_{\mathrm{LP}} \leq \frac{2b\log^2 T  \min\{\chi\theta,\delta\theta\}}{405}$ in \prop{regret 2}. Under this condition, the probability of each arm being played is lower bounded by $\frac{\min\{\chi\theta,\delta\theta\}}{90}$. Following a similar argument as in \append{stage 2}, we can conclude that with high probability after $\mathcal{O}(\log T)$ rounds, the condition $\Vert \boldsymbol{C}{\cdot,i}^L(t)-\boldsymbol{C}{\cdot,i} \Vert_\infty \leq \theta$ is satisfied.

\subsubsection{Stage 2 of Phase \uppercase\expandafter{\romannumeral2}}
\label{append:stage 2}
In Stage 2 of Phase \uppercase\expandafter{\romannumeral2}, the condition $\Vert \boldsymbol{C}_{\cdot,i}^L(t)-\boldsymbol{C}_{\cdot,i} \Vert_\infty \leq \theta$ is met. According to \prop{regret-bound}, the regret of the \algo{quant-bwk2} is determined by the term $\mathbb{E}\left[\boldsymbol{B}^{(\tau)}\right]^\top \boldsymbol{\eta}^*$, where $\boldsymbol{B}^{(t)}$ represents the vector of remaining resources at time $t$, and $\tau$ is the algorithm's stopping time. To facilitate the analysis, we define the average remaining resources as $\boldsymbol{b}^{(t)}:=\frac{\boldsymbol{B}^{(t)}}{T-t}$. We then define $\tau^\prime$ as the first time step at which $\boldsymbol{b}^{(t)}$ exits the interval $[\boldsymbol{b} - \epsilon, \boldsymbol{b} + \epsilon]$, which marks the end of Stage 2. The value of parameter $\epsilon$ is specified in \lem{part2}.
\begin{equation}
\label{eq:tau-prime}
    \tau^\prime := \min \{t: \boldsymbol{b}^{(t)} \not \in [\boldsymbol{b}-\epsilon,\boldsymbol{b}+\epsilon]\}.
\end{equation}
It is straightforward that with $\epsilon$ small enough, $\tau^\prime<\tau$. In the following analysis, we focus on bounding the expectation $\mathbb{E}\left[T-\tau^\prime \right]$ to bind the upper bound of $\mathbb{E}\left[\boldsymbol{B}^{(\tau)}\right]$. 

For each $t \in [T]$, we define the following event. Let $\mathcal{H}_{t-1} = {(\boldsymbol{r}_s, \boldsymbol{C}_s)}_{s=1}^{t-1}$ denote the filtration up to time $t$, which captures all reward and resource consumption observations up to round $t-1$. Define the event $\mathcal{E}_t$ as follows: it occurs if for all $s \in [t-1]$, the average remaining resource $\boldsymbol{b}^{(s)}$ stays within the interval $[\boldsymbol{b} - \epsilon, \boldsymbol{b} + \epsilon]$, and at time $t$, the expected resource consumption $\mathbb{E}[\boldsymbol{C}_{\cdot,i_{t}}(\boldsymbol{b}^{(t-1)})|\mathcal{H}_{t-1}]$ deviates from  $\boldsymbol{b}^{(t-1)}$ by more than $\epsilon_t$ in the $\ell_\infty$ norm. The parameter $\epsilon_t$ will be formally specified in \lem{part2}. 
\begin{equation}
    \mathcal{E}_t \coloneqq  \left\{\boldsymbol{b}^{(s)}\in[\boldsymbol{b}-\epsilon,\boldsymbol{b}+\epsilon]\text{ for each $s\in[t-1]$} \: \text{ and }\: \Vert\mathbb{E}[\boldsymbol{C}_{\cdot,i_{t}}(\boldsymbol{b}^{(t-1)})|\mathcal{H}_{t-1}]-\boldsymbol{b}^{(t-1)}\Vert_\infty > \epsilon_t \right\}.
\end{equation}

To analyze $\mathbb{E}\left[T - \tau^\prime\right]$, we examine the probability $\mathbb{P}(\tau^\prime \leq t)$, which can be decomposed into two parts. Here, the complement of the event $\mathcal{E}_t$, denoted $\mathcal{E}_t^{\complement}$, is defined over the space of $t$ vectors in $\mathbb{R}^d$.
\begin{equation}
\mathbb{P}\left(\tau^\prime\leq t\right)
 =  \mathbb{P}\left(\tau^\prime\leq t, \bigcap\limits_{s=1}^{T}\mathcal{E}_t^{\complement} \right) + \mathbb{P}\left( \tau^\prime\leq t, \bigcup\limits_{s=1}^{T}\mathcal{E}_t \right) 
 \leq \mathbb{P}\left( \tau^\prime\leq t, \bigcap\limits_{s=1}^{T}\mathcal{E}_t^{\complement} \right) + \mathbb{P}\left( \bigcup\limits_{s=1}^{T}\mathcal{E}_t \right)  .    
\end{equation}
Our analysis of the impact of the approximate LP solutions on these two probability terms relies on \lem{part2} and \lem{part1}. We present these lemmas below and defer their proofs to the end of the subsection.

\begin{lemma}
\label{lem:part2}
    Let $\alpha = \frac{\min\{1,\sigma\} \min\{\chi,\delta\}b}{d^{3/2}}$, $\epsilon=\frac{\alpha}{5}$,
    \begin{equation}
      \epsilon_t=\left\{ \begin{array}{lc}
          \frac{6}{5} & t < \frac{\alpha T}{19} \\
          \frac{19b\min\{m,d\}\sqrt{d}}{\sigma} \sqrt{\frac{\log T}{t}}+ \frac{4 \min\{m^{\frac52},d^{\frac52}\} \epsilon_{\mathrm{LP}} }{\sigma \log^2 T} & t \geq \frac{\alpha T}{19}
      \end{array} \right..
    \end{equation}
    If $\epsilon_{\mathrm{LP}} \leq \frac{\sigma \chi \log^2T}{40\min\{m^{\frac32},d^{\frac32}\}}$, then there exists a constant $\underline{T}$ such that for all $T>\underline{T}$
    \begin{equation}
        \mathbb{P}\left(
            \bigcup\limits_{s=1}^{T}\mathcal{E}_t
         \right) \leq \frac{\min\{m,d\}}{T^3}.
    \end{equation}
\end{lemma}

\begin{lemma}
\label{lem:part1}
If $\epsilon, \bar{\epsilon } , \alpha$ and $\epsilon _t$ satisfy that $\epsilon _t= \bar{\epsilon }$ for all $t \leq \alpha T$ and
\begin{equation}
\label{eq:condition}
    \frac{\alpha\bar{\epsilon}}{1-\alpha}+\sum_{t=\alpha T+1}^{T-1}\frac{\epsilon_t}{T-t}\leq\frac{2\epsilon}3,
\end{equation}
then the probability is bounded by
\begin{equation}
    \mathbb{P}\left(\tau^\prime \leq t ,\:\bigcap_{s=1}^T\mathcal{E}_s^\complement \right) 
     \leq 2d\exp \left( -\frac{ (T-t-1) \epsilon^2}{18} \right).
\end{equation}
\end{lemma}

Based on \lem{part2} and \lem{part1}, we have
\begin{equation}
\begin{aligned}
\mathbb{P}\left(\tau^\prime \leq t \right)
 \leq \mathbb{P} \left( \tau^\prime \leq t, \bigcap\limits_{s=1}^{t}\mathcal{E}_t^{\complement} \right) + \mathbb{P}\left( \bigcup\limits_{s=1}^{t}\mathcal{E}_t \right) 
\leq 2d\exp \left( -\frac{ (T-t-1) {\min\{1,\sigma^2 \} \min\{\chi^2,\delta^2 \}b^2}}{450 d^3} \right) + \frac{\min\{m,d\}}{T^3}.
\end{aligned}
\end{equation}
Consequently, the expectation can be bounded as
\begin{equation}
\begin{aligned}
    \mathbb{E}[T - \tau^\prime] & = \sum_{t=1}^T \mathbb{P}(\tau^\prime \leq t) \\
    & \leq  \sum_{t=1}^T 2d\exp \left( -\frac{ (T-t-1) {\min\{1,\sigma^2 \} \min\{\chi^2,\delta^2 \}b^2}}{450 d^3} \right) + \frac{\min\{m,d\}}{T^3} \\
    & \leq \frac{900 d^4}{ {\min\{1,\sigma^2 \} \min\{\chi^2,\delta^2 \}b^2}} + \frac{\min\{m,d\}}{T^2}.
\end{aligned}
\end{equation}
As a result,
\begin{equation}
    \mathbb{E}[\boldsymbol{B}^{(\tau)}] \leq \mathbb{E}[\boldsymbol{B}^{(\tau^\prime)}] \leq \mathbb{E} \left[ (T-\tau^\prime)(\boldsymbol{b}+\epsilon\boldsymbol{1})  \right] \leq \mathbb{E} \left[ (T-\tau^\prime)\frac{6}{5}\boldsymbol{b}  \right] = \mathcal{O}\left( 
    \frac{d^4}{ {\min\{1,\sigma^2 \} \min\{\chi^2,\delta^2 \}b}}
    \right).
\end{equation}
Using the fact that $\boldsymbol{1}^\top \boldsymbol{\eta}^* \leq \frac{1}{b}$, we derive the expected regret bound $\mathbb{E}[\boldsymbol{B}^{(\tau)}] \boldsymbol{\eta}^*  \leq \mathcal{O}\left( 
    \frac{d^4}{ {\min\{1,\sigma^2 \} \min\{\chi^2,\delta^2 \}b^2}}
    \right)$ for Phase \uppercase\expandafter{\romannumeral2} of \algo{quant-bwk2}, thereby completing the proof of \prop{regret 2}.

\paragraph{Proof of \lem{part2}} 
The key to the proof of \lem{part2} is to establish a lower bound on the number of times each optimal arm is played. To this end, we first invoke \lem{prob-arm}, which provides a lower bound on the probability of selecting each optimal arm when the LP problems are solved nearly exactly. 
\begin{lemma}
\label{lem:prob-arm}
    With $\epsilon \leq \frac{\min\{\chi,\delta\}}{5d^{3/2}}$ and \assum{warm start assum}, if $\boldsymbol{b}^{(t)} \in [\boldsymbol{b}-\epsilon,\boldsymbol{b}+\epsilon]$,  the optimal solution of LP \eq{residual-lp} at time $t+1$ is 
    \begin{equation}
        \boldsymbol{{\eta}}^*(t) = (\boldsymbol{{C}}^L(t))^{-1} \boldsymbol{{B}}^{(t)},
    \end{equation}
    and the normalized probability $\widetilde{{\boldsymbol{{\eta}}}}(t)$ used to play arm at time $t+1$ satisfies
    \begin{equation}
        \widetilde{{\boldsymbol{{\eta}}}}(t) = (\boldsymbol{{C}}^L(t))^{-1} \boldsymbol{{b}}^{(t)},
    \end{equation}
    and is element-wise lower bounded by 
    \begin{equation}
        \widetilde{{\boldsymbol{{\eta}}}}(t) \geq \frac{\chi}{5}.
    \end{equation}
\end{lemma}
In our algorithm, the LP problems are solved approximately. \lem{approx-prob-arms} demonstrates that if each LP is solved to within an additive error of at most $\epsilon_{\mathrm{LP}} (T - t)/\log^2 T$, and $\epsilon_{\mathrm{LP}} \leq \frac{\sigma \chi \log^2 T}{40\min\{m^{3/2}, d^{3/2}\}}$, then the resulting approximate solution remains close to the exact one in terms of the normalized probability.
\begin{lemma}
\label{lem:approx-prob-arms}
    If the LP problem is solved $\epsilon_{\mathrm{LP}} (T-t)/\log^2 T$ optimal and $\epsilon_{\mathrm{LP}} (T-t)/\log^2 T$-approximate, let the approximate solution to the LP be denoted by ${\boldsymbol{\eta}^*}^\prime(t)$, and define the probability of playing each arm as $\hat{\boldsymbol{\eta}}(t) = {\boldsymbol{\eta}^*}^\prime(t) / \vert {\boldsymbol{\eta}^*}^\prime (t)\vert_1 $. Then, if $ \epsilon_{\mathrm{LP}} \leq \frac{\sigma \chi \log^2T}{40\min\{m^{\frac32},d^{\frac32}\}}$, the element-wise difference in the probabilities of playing each arm between the approximate and exact solutions to the LP is bounded by
     \begin{equation}
        \Vert \tilde{\boldsymbol{\eta}}(t) - \hat{\boldsymbol{\eta}}(t) \Vert_\infty \leq \frac{4 \min\{m^{\frac32},d^{\frac32}\} \epsilon_{\mathrm{LP}} }{\sigma \log^2 T} \leq \frac{\chi}{10}.
    \end{equation}
\end{lemma}
\begin{proof}
    By obtaining an approximate solution to the LP,
    \begin{equation}
         \Vert \boldsymbol{{C}}^L(t) ({\boldsymbol{\eta}^*}^\prime(t)-{\boldsymbol{\eta}^*}(t)) \Vert_\infty \leq \epsilon_{\mathrm{LP}} (T-t)/\log^2 T 
    \end{equation}
    then
    \begin{equation}
    \begin{aligned}
        \Vert {\boldsymbol{\eta}^*}^\prime(t) -   {\boldsymbol{\eta}^*}(t) \Vert_\infty &\leq \Vert {\boldsymbol{\eta}^*}^\prime(t) -   {\boldsymbol{\eta}^*}(t) \Vert_2 \\
        & \leq \frac{1}{\sigma_{\min}(\boldsymbol{{C}}^L(t))} \Vert \boldsymbol{{C}}^L(t) ({\boldsymbol{\eta}^*}^\prime(t)-{\boldsymbol{\eta}^*}(t)) \Vert_2 \\
        & \leq \frac{\sqrt{\min\{m,d\}}}{\sigma_{\min}(\boldsymbol{{C}}^L(t))} \Vert \boldsymbol{{C}}^L(t) ({\boldsymbol{\eta}^*}^\prime(t)-{\boldsymbol{\eta}^*}(t)) \Vert_\infty \\
        & \leq \frac{2\sqrt{\min\{m,d\}} \epsilon_{\mathrm{LP}} (T-t)}{\sigma \log^2 T},
    \end{aligned}
    \end{equation}
    and the 1-norm of ${\boldsymbol{\eta}^*}^\prime(t)$ satisfies $ \vert \Vert{\boldsymbol{\eta}^*}^\prime(t)\Vert_1 - \Vert{\boldsymbol{\eta}^*}(t)\Vert_1 \vert \leq \min\{m,d\} \Vert {\boldsymbol{\eta}^*}^\prime(t) -   {\boldsymbol{\eta}^*}(t) \Vert_\infty$. 
    Therefore we can bound the difference in the  probabilities of playing each arm:
    \begin{equation}
    \begin{aligned}
        \Vert \tilde{\boldsymbol{\eta}}(t) - \hat{\boldsymbol{\eta}}(t) \Vert_\infty  &= \Vert \frac{{\boldsymbol{\eta}^*}^\prime(t)}{\Vert {\boldsymbol{\eta}^*}^\prime(t) \Vert_1} - \frac{{\boldsymbol{\eta}^*}(t)}{\Vert {\boldsymbol{\eta}^*}(t) \Vert_1} \Vert_\infty  \\
        & \leq  \frac{ (T-t)\Vert {\boldsymbol{\eta}^*}^\prime(t) -{\boldsymbol{\eta}^*}(t)   \Vert_\infty } {(T-t)\Vert {\boldsymbol{\eta}^*}^\prime(t) \Vert_1} + \frac{\Vert {\boldsymbol{\eta}^*} (t) \Vert_\infty  \left( T-t- \Vert {\boldsymbol{\eta}^*}^\prime(t) \Vert_1  
         \right) }{(T-t)\Vert {\boldsymbol{\eta}^*}^\prime(t) \Vert_1} \\
         & \leq \frac{(1+\min\{m,d\}) \frac{2 \min\{m^{\frac12},d^{\frac12}\} \epsilon_{\mathrm{LP}} }{\sigma \log^2 T} }{1-\frac{2 \min\{m^{\frac32},d^{\frac32}\} \epsilon_{\mathrm{LP}} }{\sigma \log^2 T}} \\
        & \leq \frac{4 \min\{m^{\frac32},d^{\frac32}\} \epsilon_{\mathrm{LP}} }{\sigma \log^2 T},
    \end{aligned}
    \end{equation}
    where we apply $\frac{2 \min\{m^{\frac32},d^{\frac32}\} \epsilon_{\mathrm{LP}} }{\sigma \log^2 T} \leq \frac{\chi}{20} < \frac12 $ in the last line.
\end{proof}
From \lem{prob-arm}, we know that if the LP problem is solved approximately, the approximate probability is element-wise lower bounded by $\hat{\boldsymbol{\eta}}(t) \geq \tilde{\boldsymbol{\eta}}(t) -\Vert \tilde{\boldsymbol{\eta}}(t) - \hat{\boldsymbol{\eta}}(t) \Vert_\infty\geq \frac{\chi}{10}$. 
This implies that each optimal arm is selected with probability at least $\frac{\chi}{10}$. In the following, we show that if $\boldsymbol{b}^{(s)} \in [\boldsymbol{b} - \epsilon, \boldsymbol{b} + \epsilon]$ for all $s \in [t]$ and $t \geq \frac{3200 \log T}{\chi^2}$, then the number of times arm $i\in \mathcal{I}^*$ is played up to time $t$, denoted as $n_i(t)$, is lower bounded by $n_i(t) \geq \frac{\chi t}{20}$.

Recall that the filtration $\mathcal{H}_t$ contains all random rewards and resource consumption observations up to time $t$, i.e., $\mathcal{H}_t = \{\boldsymbol{r}_s,\boldsymbol{C}_s\}_{s=1}^{t}$. Define the martingale difference sequence $X_s = n_i(s)-n_i(s-1)-\hat{\eta}_i(s)$. By the Azuma–Hoeffding inequality, for all $t \in [T]$, we have
\begin{equation}
    \mathbb{P}\left(\sum_{s=1}^t n_i(s)-n_i(s-1)-\hat{\eta}_i(s) \leq -\sqrt{8t\log T}\right) \leq \frac{1}{T^4}.
\end{equation}
This leads to 
\begin{equation}
    \mathbb{P}\left(n_i(t) \leq \sum_{s=1}^t\hat{\eta}_i(s) -\sqrt{8t\log T},\text{ for some $t\in[T]$}\right) \leq \frac{1}{T^3}.
\end{equation}
Therefore, we know that $n_i(t) \geq \frac{\chi t}{20}$ holds with probability $1-\frac{1}{T^3}$ if $\boldsymbol{b}^{(s)}\in [\boldsymbol{b}-\epsilon,\boldsymbol{b}+\epsilon]$ for all $s\in [t]$ and $t\geq \frac{3200\log T}{\chi^2}$, and the bias in the resource consumption can be upper bounded as 
\begin{equation}
\begin{aligned}
\label{eq:gap-e-b}
\Vert \mathbb{E} \left[\boldsymbol{C}_{\cdot, i_{t+1}}(\boldsymbol{b}^{(t)}) | \mathcal{H}_t\right] -\boldsymbol{b}^{(t)}\Vert_\infty & =\Vert\boldsymbol{C}\hat{\boldsymbol{\eta}}(t)-\boldsymbol{b}^{(t)}\Vert_\infty \\
& \leq \Vert\boldsymbol{C}\tilde{\boldsymbol{\eta}}(t)  -\boldsymbol{b}^{(t)}\Vert_\infty + \Vert \boldsymbol{C}\Vert_\infty  \Vert(\tilde{\boldsymbol{\eta}}(t) - \hat{\boldsymbol{\eta}}(t)) \Vert_\infty \\
 & \leq\Vert\left(\boldsymbol{C}-\boldsymbol{C}^L(t)\right)\left(\boldsymbol{C}^L(t)\right)^{-1}\boldsymbol{b}^{(t)}\Vert_\infty + \frac{4 \min\{m^{\frac52},d^{\frac52}\} \epsilon_{\mathrm{LP}} }{\sigma \log^2 T} \\
 & \leq\Vert\boldsymbol{C}-\boldsymbol{C}^L(t)\Vert_\infty\Vert\left(\boldsymbol{C}^L(t)\right)^{-1}\Vert_\infty\Vert\boldsymbol{b}^{(t)}\Vert_\infty +  \frac{4 \min\{m^{\frac52},d^{\frac52}\} \epsilon_{\mathrm{LP}} }{\sigma \log^2 T} \\
 & \leq \min\{m,d\} \sqrt{\frac{40  \log T}{\chi t}} \cdot \frac{2\sqrt{d}}{\sigma} \cdot{} \frac{6b}{5} + \frac{4 \min\{m^{\frac52},d^{\frac52}\} \epsilon_{\mathrm{LP}} }{\sigma \log^2 T} \\
 & \leq \frac{16b\min\{m,d\}\sqrt{d}}{\sigma} \sqrt{\frac{\log T}{t}}+ \frac{4 \min\{m^{\frac52},d^{\frac52}\} \epsilon_{\mathrm{LP}} }{\sigma \log^2 T}.
\end{aligned}
\end{equation}
The last line is defined to be $\epsilon_t$ when $t > \frac{\alpha T}{19}$. Otherwise, when $t \leq \frac{\alpha T}{19}$, $\epsilon_t$ is defined to be $\frac{6}{5}$.  Given $\alpha$, consider the minimum $\underline{T}$ that satisfies $\frac{\alpha T}{19}  > \frac{3200\log T}{\chi^2}$. 
For any $T > \underline{dT}$, we have the following two cases:
\begin{enumerate}
    \item $t \geq \frac{\alpha T}{19}$. In this case, we upper bound the gap between the expected resource consumption and the average remaining resource as shown in \eq{gap-e-b}. From the definition of $\epsilon_t$, we have $\Vert \mathbb{E} \left[\boldsymbol{C}_{\cdot,i_{t+1}}(\boldsymbol{b}^{(t)}) | \mathcal{H}_t\right] -\boldsymbol{b}^{(t)}\Vert_\infty \leq \epsilon_t$, which implies that the event $\mathcal{E}_t$ is contained in the complement of the following set:
    \begin{equation}
        \mathcal{E}_t \subset \left\{ n_i(t) \geq \sum_{s=1}^t\hat{\eta}_i(s) -\sqrt{8t\log T},\text{ for all $i \in [m]$}\right\}^{\complement}.
    \end{equation}
    \item  $t < \frac{\alpha T}{19}$. Since the resource consumption matrix is element-wise bounded in $[0,1]$, we have $\boldsymbol{0}\leq  \mathbb{E} \left[\boldsymbol{C}_{\cdot,i_{t+1}}(\boldsymbol{b}^{(t)}) | \mathcal{H}_t\right] \leq \boldsymbol{1}$. In addition, $\boldsymbol{0} \leq \boldsymbol{b}^{(t)} \leq \boldsymbol{\frac{6}{5}}$. As a result, $ \Vert \mathbb{E} \left[\boldsymbol{C}_{\cdot,i_{t+1}}(\boldsymbol{b}^{(t)}) | \mathcal{H}_t\right] -\boldsymbol{b}^{(t)}\Vert_\infty \leq \frac{6}{5} = \epsilon_t$. Since the event  $\mathcal{E}_t$ requires $\Vert \mathbb{E} \left[\boldsymbol{C}_{\cdot,i_{t+1}}(\boldsymbol{b}^{(t)}) | \mathcal{H}_t\right] -\boldsymbol{b}^{(t)}\Vert_\infty$ to be greater than $\epsilon_t$, it follows that $\mathcal{E}_t = \emptyset$ in this case.
\end{enumerate}

Combining the two cases, we obtain the following bound, which completes the proof of \lem{part2}:
\begin{equation}
    \mathbb{P}\left(\bigcup_{t=1}^T \mathcal{E}_t\right) \subset \mathbb{P}\left( \left\{ n_i(t) \geq \sum_{s=1}^t\hat{\eta}_i(s) -\sqrt{8t\log T},\text{ for all $t\in [T]$ and all $i \in [m]$}\right\}^{\complement} \right) \leq \frac{\min\{m,d\}}{T^3}.
\end{equation}

\paragraph{Proof of \lem{part1}}
In the proof of \lem{part1}, our argument deviates from that of \cite{li2021symmetry} by employing a different choice of $\epsilon_t$, which reflects the approximation error incurred when solving the LP problems. For any $j \in [d]$, define
\begin{equation}
    \tau^\prime_j = \min \{t: b_j^{(s)} \not \in [b-\epsilon,b+\epsilon]\}.
\end{equation}
That is, $\tau^\prime_j$ represents the first time the $j$-th entry of $ \boldsymbol{b}^{(s)}$ exceeds the interval $[\boldsymbol{b}-\epsilon,\boldsymbol{b}+\epsilon]$. Recall that $\tau^\prime = \min \{t: \boldsymbol{b}^{(t)} \not \in [\boldsymbol{b}-\epsilon,\boldsymbol{b}+\epsilon]\}$, which implies that $\tau^\prime = \min_j \{\tau^\prime_j\}$. Consequently, the following event holds:
\begin{equation}
    \left\{\tau^\prime \leq t ,\:\bigcap_{s=1}^t\mathcal{E}_s^\complement \right\} = \bigcup_{j\in [d]} \left\{ \tau_j^\prime \leq t,\: \tau_k^\prime \geq \tau_j^\prime \text{ for $k \in [j]$ and $k \neq j$ },\:\bigcap_{s=1}^t\mathcal{E}_s^\complement \right\}.
\end{equation}
With $\mathcal{F}_j = \left\{\tau_k^\prime \geq \tau_j^\prime \text{ for $k \in [j]$ and $k \neq j$} \right\}$, we obtain the following probability bound:
\begin{equation}
    \mathbb{P}\left(\tau^\prime \leq t ,\:\bigcap_{s=1}^t\mathcal{E}_s^\complement \right) \leq \sum_{j\in [d]} \mathbb{P}\left( \tau_j^\prime \leq t,\: \mathcal{F}_j,\:\bigcap_{s=1}^t\mathcal{E}_s^\complement \right).
\end{equation}

Next, we derive a bound for the probability corresponding to each constraint $j\in[d] $. Given the original sequence $b_j^{(s)}$ We define the sequence $\tilde{b}_j^{(s)}$ as follows: 
\begin{equation}
    \tilde{b}_j^{(s)} =\begin{cases} 
    b_j^{(s)} & s \leq  \tau^\prime_j \\
    \tilde{b}_j^{(s-1)} & s > \tau^\prime_j
    \end{cases}.
\end{equation}
It is straightforward to see that $\tilde{b}_j^{(s)}$ can be interpreted as a truncated version of $b_j^{(s)}$ when it exceeds the threshold $\epsilon$. Notably, we have the following equivalence:
\begin{equation}
\hspace{-2mm} \left\{{b}^{(s)}_j\not\in[b-\epsilon,b+\epsilon] \text{ for some $s\leq t$, } \mathcal{F}_j,\:\bigcap\limits_{s=1}^{t}\mathcal{E}_t^{\complement}\right\}
        =
        \left\{
            \tilde{{b}}^{(s)}_j\not\in[b-\epsilon,b+\epsilon] \text{ for some } s\leq t,\: \mathcal{F}_j,\:\bigcap\limits_{s=1}^{t}\mathcal{E}_t^{\complement}\right\}.
\end{equation}

The definition of $\tilde{{b}}^{(s)}_j$ has the following properties. Consider a sample point 
\begin{equation}
    \omega = \{ \tilde{{b}}^{(0)}_j, \tilde{{b}}^{(1)}_j,\dots,\tilde{{b}}^{(t)}_j \}\in \left\{\tilde{{b}}^{(s)}_j\not\in[b-\epsilon,b+\epsilon] \text{ for some } s\leq t, \: \mathcal{F}_j,\:\bigcap\limits_{s=1}^{t}\mathcal{E}_t^{\complement}\right\},
\end{equation}
where $\omega$ is a sequence satisfying
\begin{enumerate}
    \item $\tilde{{b}}^{(s)}_j \in [b-\epsilon, b+\epsilon] \text{, for $s < \tau^\prime_j $}$,
    \item $ \tilde{{b}}^{(s)}_j \equiv \tilde{{b}}^{(\tau^\prime_j)}_j \not \in [b-\epsilon, b+\epsilon] \text{, for $ \tau^\prime_j \leq s \leq t$} $,
    \item $\tilde{{b}}^{(s)}_k \in [b-\epsilon, b+\epsilon]$  for $s \leq \tau^\prime_j$ and $k\in[d],\: k \neq j$.
\end{enumerate}

Then from the definition of $\tilde{{b}}^{(s)}_j$ and ${b}^{(s)}_j$, the increment $\Delta_s = \tilde{{b}}_{j}^{(s)}-\tilde{{b}}_{j}^{(s-1)}$ is given by
\begin{equation}
\begin{aligned}
        \Delta _s
        &=\tilde{{b}}_{j}^{(s)}-\tilde{{b}}_{j}^{(s-1)}
        =
        \left\{
            \begin{matrix}
                0, & s > \tau^\prime_j,\\ -\frac{1}{T-s}(C_{j,s}-\tilde{{b}}_{j}^{(s-1)}),   & s \leq \tau^\prime_j .
            \end{matrix}
        \right.
\end{aligned}
\end{equation}
 With $\tilde{{b}}^{(0)}_j = b$, we obtain
\begin{equation}
\label{eq:summation break}
 \tilde{{b}}^{(s)}_j-b = \sum_{k=1}^s \Delta_k=\sum_{k=1}^s \mathbb{E}[\Delta_k|\mathcal{H}_{k-1}] + \sum_{k=1}^s \left(\Delta_k -\mathbb{E}[\Delta_k|\mathcal{H}_{k-1}]\right).
\end{equation}
From the definition of $\bigcap\limits_{s=1}^{t}\mathcal{E}_t^{\complement}$ and the properties of the sequence $\tilde{b}_j^{(s)}$, we observe that:
\begin{enumerate}
    \item Since $\{\tilde{{b}}^{(s)}_j \in [b-\epsilon, b+\epsilon] \text{, for $s < \tau^\prime_j$}\}$ and $\left\{\tau_k^\prime \geq \tau_j^\prime \text{ for $k \in [j]$ and $k \neq j$} \right\}$, we have $\mathbb{E}[\Delta_s|\mathcal{H}_{s-1}] \leq \frac{\epsilon_s}{T-s} \text{, for $s \leq \tau^\prime_j$}$.
    \item Since $\{\tilde{{b}}^{(s)}_j \equiv \tilde{{b}}^{(s-1)}_j  \text{, for $s > \tau^\prime_j$}\}$, it follows that $\mathbb{E}[\Delta_s|\mathcal{H}_{s-1}] \equiv 0\text{ for $s > \tau^\prime_j$}$. 
\end{enumerate}

For the first summation in \eq{summation break}, it is bound by 
\begin{equation}
\begin{aligned}
   \left\vert \sum_{k=1}^s \mathbb{E}[\Delta_k|\mathcal{H}_{k-1}] \right\vert 
    \leq  \sum_{k=1}^T \left\vert \mathbb{E}[\Delta_k|\mathcal{H}_{k-1}] \right\vert 
    = \sum_{k=1}^{\alpha T} \frac{\bar{\epsilon}}{T-k} + \sum_{k=\alpha T+1}^{T-1} \frac{\epsilon_k} {T-k} 
    \leq  \frac{\alpha \bar{\epsilon}}{1-\alpha} + \sum_{k=\alpha T+1}^{T-1} \frac{\epsilon_k} {T-k}.
\end{aligned}
\end{equation}
Consequently, if condition \eqref{eq:condition} holds, namely, $\frac{\alpha\bar{\epsilon}}{1-\alpha}+\sum_{t=\alpha T+1}^{T-1}\frac{\epsilon_t}{T-t}\leq\frac{2\epsilon}3$, 
we have $\left\vert \sum_{k=1}^s \mathbb{E}[\Delta_k|\mathcal{H}_{k-1}] \right\vert \leq \frac{2\epsilon}{3}$. 

For the second summation in \eq{summation break}, the term  $\Delta_k -\mathbb{E}[\Delta_k|\mathcal{H}_{k-1}]$ forms a martingale difference sequence. Moreover, since $C_{j,t}\in[0,1]$ and $\mathbb{E}[C_{j,t}|\mathcal{H}_{k-1}] \in [0,1]$, we obtain the bound
\begin{equation}
    \left\vert \Delta_k -\mathbb{E}[\Delta_k|\mathcal{H}_{k-1}] \right\vert  = \frac{1}{T-k} \left\vert C_{j,t} - \tilde{b}_j^{(t-1)} - \left(\mathbb{E}[C_{j,t}|\mathcal{H}_{k-1}]  - \tilde{b}_j^{(t-1)} \right) \right\vert \leq \frac{1}{T-k}.
\end{equation}
Applying the Azuma–Hoeffding inequality, we derive
\begin{equation}
    \mathbb{P}\left( \left\vert \sum_{k=1}^s \Delta_k -\mathbb{E}[\Delta_k|\mathcal{H}_{k-1}] \right\vert \geq \frac{\epsilon}{3} \text{, for some $s \leq t$} ,\:\mathcal{F}_j,\:\bigcap\limits_{s=1}^{t}\mathcal{E}_t^{\complement} \right) \leq 2\exp \left( -\frac{(\frac{\epsilon}{3})^2}{2\sum_{k=1}^t \frac{1}{(T-k)^2}} \right).
\end{equation}
To further simplify the probability, we note that $\sum_{k=1}^t \frac{1}{(T-k)^2} \leq \sum_{k=1}^t \frac{1}{(T-k)(T-k-1)} \leq \frac{1}{T-t-1}$. 
Substituting this result, we obtain
\begin{equation}
     \mathbb{P}\left( \left\vert \sum_{k=1}^s \Delta_k -\mathbb{E}[\Delta_k|\mathcal{H}_{k-1}] \right\vert \geq \frac{\epsilon}{3} \text{, for some $s \leq t$} ,\:\mathcal{F}_j,\:\bigcap\limits_{s=1}^{t}\mathcal{E}_t^{\complement} \right) \leq 2\exp \left( -\frac{ (T-t-1) \epsilon^2}{18} \right).
\end{equation}

Finally, the event $\left\{\tilde{\tau}^\prime_j \leq t ,\:\mathcal{F}_j,\:\bigcap\limits_{s=1}^{t}\mathcal{E}_t^{\complement}\right\}$ can be appropriately bounded as follows:
\begin{small}
\begin{align}
\left\{\tilde{\tau}^\prime_j \leq t ,\:\mathcal{F}_j,\:\bigcap\limits_{s=1}^{t}\mathcal{E}_t^{\complement}\right\}
=&\left\{{b}^{(s)}_j\not\in[b-\epsilon,b+\epsilon] \text{ for some $s\leq t$, } \mathcal{F}_j,\:\bigcap\limits_{s=1}^{t}\mathcal{E}_t^{\complement}\right\}\nonumber \\
    \subset&
    \left\{
        \left|\sum\limits_{k=1}^{s} \Delta_k\right|\geq\epsilon \text{ for some }s\leq t ,\:\mathcal{F}_j,\:\bigcap\limits_{s=1}^{t}\mathcal{E}_t^{\complement}
    \right\}\nonumber\\
    \subset&
    \left\{
        \left\vert\sum_{k=1}^s \left(\Delta_k-\mathbb{E}[\Delta_k|\mathcal{H}_{k-1}]\right) \right\vert \geq \frac{\epsilon}{3} \text{ for some } s\le t,\:\mathcal{F}_j,\:\bigcap\limits_{s=1}^{t}\mathcal{E}_t^{\complement}
    \right\}\bigcup\left\{
        \left\vert
            \sum\limits_{k=1}^{s} \mathbb{E}[\Delta_k|\mathcal{H}_{k-1}]
        \right\vert
        >
        \frac{2\epsilon}{3} ,\:\mathcal{F}_j,\:\bigcap\limits_{s=1}^{t}\mathcal{E}_t^{\complement} \right\}\nonumber \\
         =& \left\{
        \left\vert\sum_{k=1}^s \left(\Delta_k-\mathbb{E}[\Delta_k|\mathcal{H}_{k-1}]\right) \right\vert \geq \frac{\epsilon}{3} \text{ for some } s \leq t ,\:\mathcal{F}_j,\:\bigcap\limits_{s=1}^{t}\mathcal{E}_t^{\complement} \right\},
\end{align}
\end{small} 
which indicates that
\begin{equation}
\begin{aligned}
    \mathbb{P}\left(\tau^\prime \leq t ,\:\bigcap_{s=1}^t\mathcal{E}_s^\complement \right) 
    & \leq \sum_{j\in [d]} \mathbb{P}\left( \tau_j^\prime \leq t,\: \mathcal{F}_j,\:\bigcap_{s=1}^t\mathcal{E}_s^\complement \right) \\
    & = \sum_{j\in [d]} \mathbb{P}\left( \left\vert \sum_{k=1}^s \Delta_k -\mathbb{E}[\Delta_k|\mathcal{H}_{k-1}] \right\vert \geq \frac{\epsilon}{3} \text{, for some $s \leq t$} ,\:\mathcal{F}_j,\:\bigcap\limits_{s=1}^{t}\mathcal{E}_t^{\complement} \right)  \\
    & \leq 2d\exp \left( -\frac{ (T-t-1) \epsilon^2}{18} \right).
\end{aligned}
\end{equation}


\begin{thebibliography}{10}

\bibitem{Abbas2024}
Amira Abbas, Andris Ambainis, Brandon Augustino, Andreas Bärtschi, Harry Buhrman, Carleton Coffrin, Giorgio Cortiana, Vedran Dunjko, Daniel~J. Egger, Bruce~G. Elmegreen, Nicola Franco, Filippo Fratini, Bryce Fuller, Julien Gacon, Constantin Gonciulea, and Christa Zoufal, \emph{Challenges and opportunities in quantum optimization}, Nature Reviews Physics \textbf{6} (2024), no.~12, 718--735.

\bibitem{agrawal2014dynamic}
Shipra Agrawal, Zizhuo Wang, and Yinyu Ye, \emph{A dynamic near-optimal algorithm for online linear programming}, Operations Research \textbf{62} (2014), no.~4, 876--890, \arxiv{0911.2974}.

\bibitem{alman2025more}
Josh Alman, Ran Duan, Virginia~Vassilevska Williams, Yinzhan Xu, Zixuan Xu, and Renfei Zhou, \emph{More asymmetry yields faster matrix multiplication}, Proceedings of the 2025 Annual ACM-SIAM Symposium on Discrete Algorithms (SODA), pp.~2005--2039, SIAM, 2025, \arxiv{2404.16349}.

\bibitem{vanApeldoorn2018SDP}
Joran~van Apeldoorn and Andr{\'a}s Gily{\'e}n, \emph{Improvements in quantum {SDP}-solving with applications}, Proceedings of the 46th International Colloquium on Automata, Languages, and Programming, Leibniz International Proceedings in Informatics (LIPIcs), vol. 132, pp.~99:1--99:15, Schloss Dagstuhl--Leibniz-Zentrum fuer Informatik, 2019, \arxiv{1804.05058}.

\bibitem{van2019quantum}
Joran~van Apeldoorn and Andr{\'a}s Gily{\'e}n, \emph{Quantum algorithms for zero-sum games}, arXiv preprint arXiv:1904.03180 (2019), \arxiv{1904.03180}.

\bibitem{van2017quantum}
Joran~van Apeldoorn, Andr{\'a}s Gily{\'e}n, Sander Gribling, and Ronald de~Wolf, \emph{Quantum {SDP}-solvers: Better upper and lower bounds}, 2017 IEEE 58th Annual Symposium on Foundations of Computer Science (FOCS), pp.~403--414, IEEE, 2017, \arxiv{1705.01843}.

\bibitem{van2020convex}
Joran~van Apeldoorn, Andr{\'a}s Gily{\'e}n, Sander Gribling, and Ronald~de Wolf, \emph{Convex optimization using quantum oracles}, Quantum \textbf{4} (2020), 220, \arxiv{1809.00643}.

\bibitem{auer2002nonstochastic}
Peter Auer, Nicolo Cesa-Bianchi, Yoav Freund, and Robert~E. Schapire, \emph{The nonstochastic multiarmed bandit problem}, SIAM Journal on Computing \textbf{32} (2002), no.~1, 48--77.

\bibitem{augustino2025fast}
Brandon Augustino, Dylan Herman, Enrico Fontana, Junhyung~Lyle Kim, Jacob Watkins, Shouvanik Chakrabarti, and Marco Pistoia, \emph{Fast convex optimization with quantum gradient methods}, arXiv preprint arXiv:2503.17356 (2025), \arxiv{2503.17356}.

\bibitem{augustino2023quantum}
Brandon Augustino, Jiaqi Leng, Giacomo Nannicini, Tam{\'a}s Terlaky, and Xiaodi Wu, \emph{A quantum central path algorithm for linear optimization}, arXiv preprint arXiv:2311.03977 (2023), \arxiv{2311.03977}.

\bibitem{augustino2023interior}
Brandon Augustino, Giacomo Nannicini, Tam{\'a}s Terlaky, and Luis~F. Zuluaga, \emph{Quantum interior point methods for semidefinite optimization}, Quantum \textbf{7} (2023), 1110, \arxiv{2112.06025}.

\bibitem{badanidiyuru2018bandits}
Ashwinkumar Badanidiyuru, Robert Kleinberg, and Aleksandrs Slivkins, \emph{Bandits with knapsacks}, Journal of the ACM (JACM) \textbf{65} (2018), no.~3, 1--55, \arxiv{1305.2545}.

\bibitem{bertsimas2000restless}
Dimitris Bertsimas and Jos{\'e} Ni{\~n}o-Mora, \emph{Restless bandits, linear programming relaxations, and a primal-dual index heuristic}, Operations Research \textbf{48} (2000), no.~1, 80--90.

\bibitem{besbes2012blind}
Omar Besbes and Assaf Zeevi, \emph{Blind network revenue management}, Operations Research \textbf{60} (2012), no.~6, 1537--1550.

\bibitem{borkar2002risk}
Vivek~S. Borkar and Sean~P. Meyn, \emph{Risk-sensitive optimal control for {M}arkov decision processes with monotone cost}, Mathematics of Operations Research \textbf{27} (2002), no.~1, 192--209.

\bibitem{bouland2023quantum}
Adam Bouland, Yosheb~M. Getachew, Yujia Jin, Aaron Sidford, and Kevin Tian, \emph{Quantum speedups for zero-sum games via improved dynamic {G}ibbs sampling}, International Conference on Machine Learning, pp.~2932--2952, PMLR, 2023, \arxiv{2301.03763}.

\bibitem{boyd2004convex}
Stephen~P. Boyd and Lieven Vandenberghe, \emph{Convex optimization}, Cambridge University Press, 2004.

\bibitem{van2021minimum}
Jan van~den Brand, Yin~Tat Lee, Yang~P. Liu, Thatchaphol Saranurak, Aaron Sidford, Zhao Song, and Di~Wang, \emph{Minimum cost flows, {MDP}s, and $\ell_{1}$-regression in nearly linear time for dense instances}, Proceedings of the 53rd Annual ACM SIGACT Symposium on Theory of Computing, pp.~859--869, 2021, \arxiv{2101.05719}.

\bibitem{brandao2017quantum}
Fernando G. S.~L. Brand\~{a}o, Amir Kalev, Tongyang Li, Cedric Yen-Yu Lin, Krysta~M. Svore, and Xiaodi Wu, \emph{Quantum {SDP} solvers: Large speed-ups, optimality, and applications to quantum learning}, 46th International Colloquium on Automata, Languages, and Programming (ICALP 2019) (Dagstuhl, Germany) (Christel Baier, Ioannis Chatzigiannakis, Paola Flocchini, and Stefano Leonardi, eds.), Leibniz International Proceedings in Informatics (LIPIcs), vol. 132, pp.~27:1--27:14, Schloss Dagstuhl -- Leibniz-Zentrum f{\"u}r Informatik, 2019, \arxiv{1710.02581}.

\bibitem{brandao2016quantum}
Fernando~G.S.L. Brand{\~a}o and Krysta Svore, \emph{Quantum speed-ups for semidefinite programming}, Proceedings of the 58th Annual Symposium on Foundations of Computer Science, pp.~415--426, 2017, \arxiv{1609.05537}.

\bibitem{brassard2000quantum}
Gilles Brassard, Peter H{\o}yer, Michele Mosca, and Alain Tapp, \emph{Quantum amplitude amplification and estimation}, Contemporary Mathematics \textbf{305} (2002), 53--74, \arxiv{quant-ph/0005055}.

\bibitem{chakrabarti2020quantum}
Shouvanik Chakrabarti, Andrew~M. Childs, Tongyang Li, and Xiaodi Wu, \emph{Quantum algorithms and lower bounds for convex optimization}, Quantum \textbf{4} (2020), 221, \arxiv{1809.01731}.

\bibitem{chen2025quantum}
Zherui Chen, Yuchen Lu, Hao Wang, Yizhou Liu, and Tongyang Li, \emph{Quantum {L}angevin dynamics for optimization}, Communications in Mathematical Physics \textbf{406} (2025), no.~3, 52, \arxiv{2311.15587}.

\bibitem{cohen2021solving}
Michael~B. Cohen, Yin~Tat Lee, and Zhao Song, \emph{Solving linear programs in the current matrix multiplication time}, Journal of the ACM (JACM) \textbf{68} (2021), no.~1, 1--39, \arxiv{1810.07896}.

\bibitem{cornelissen2022near}
Arjan Cornelissen, Yassine Hamoudi, and Sofiene Jerbi, \emph{Near-optimal quantum algorithms for multivariate mean estimation}, Proceedings of the 54th Annual ACM SIGACT Symposium on Theory of Computing, pp.~33--43, 2022, \arxiv{2111.09787}.

\bibitem{dai2023quantum}
Zhongxiang Dai, Gregory Kang~Ruey Lau, Arun Verma, Yao Shu, Bryan Kian~Hsiang Low, and Patrick Jaillet, \emph{Quantum {B}ayesian optimization}, Advances in Neural Information Processing Systems \textbf{36} (2023), 20179--20207, \arxiv{2310.05373}.

\bibitem{dantzig1997simplex}
George~B. Dantzig and Mukund~N. Thapa, \emph{The simplex method}, Springer, 1997.

\bibitem{dunjko2018machine}
Vedran Dunjko and Hans~J. Briegel, \emph{Machine learning \& artificial intelligence in the quantum domain: a review of recent progress}, Reports on Progress in Physics \textbf{81} (2018), no.~7, 074001.

\bibitem{ferreira2018online}
Kris~Johnson Ferreira, David Simchi-Levi, and He~Wang, \emph{Online network revenue management using thompson sampling}, Operations Research \textbf{66} (2018), no.~6, 1586--1602.

\bibitem{freund1997decision}
Yoav Freund and Robert~E. Schapire, \emph{A decision-theoretic generalization of on-line learning and an application to boosting}, Journal of Computer and System Sciences \textbf{55} (1997), no.~1, 119--139.

\bibitem{gao2024logarithmic}
Minbo Gao, Zhengfeng Ji, Tongyang Li, and Qisheng Wang, \emph{Logarithmic-regret quantum learning algorithms for zero-sum games}, Advances in Neural Information Processing Systems \textbf{36} (2024), \arxiv{2304.14197}.

\bibitem{giovannetti2008quantum}
Vittorio Giovannetti, Seth Lloyd, and Lorenzo Maccone, \emph{Quantum random access memory}, Physical Review Letters \textbf{100} (2008), no.~16, 160501, \arxiv{0708.1879}.

\bibitem{gittins2011multi}
John Gittins, Kevin Glazebrook, and Richard Weber, \emph{Multi-armed bandit allocation indices}, John Wiley \& Sons, 2011.

\bibitem{gong2022robustness}
Weiyuan Gong, Chenyi Zhang, and Tongyang Li, \emph{Robustness of quantum algorithms for nonconvex optimization}, The Thirteenth International Conference on Learning Representations, 2025, \arxiv{2212.02548}.

\bibitem{grigoriadis1995sublinear}
Michael~D. Grigoriadis and Leonid~G. Khachiyan, \emph{A sublinear-time randomized approximation algorithm for matrix games}, Operations Research Letters \textbf{18} (1995), no.~2, 53--58.

\bibitem{grover2002creating}
Lov Grover and Terry Rudolph, \emph{Creating superpositions that correspond to efficiently integrable probability distributions}, arXiv preprint quant-ph/0208112 (2002), \arxiv{quant-ph/0208112}.

\bibitem{jerbi2022quantum}
Sofiene Jerbi, Arjan Cornelissen, M{\=a}ris Ozols, and Vedran Dunjko, \emph{Quantum policy gradient algorithms}, arXiv preprint arXiv:2212.09328 (2022), \arxiv{2212.09328}.

\bibitem{jiang2020faster}
Shunhua Jiang, Zhao Song, Omri Weinstein, and Hengjie Zhang, \emph{Faster dynamic matrix inverse for faster {LP}s}, arXiv preprint arXiv:2004.07470 (2020), \arxiv{2004.07470}.

\bibitem{karmarkar1984new}
Narendra Karmarkar, \emph{A new polynomial-time algorithm for linear programming}, Proceedings of the Sixteenth Annual ACM Symposium on Theory of Computing, pp.~302--311, 1984.

\bibitem{kerenidis2016quantum}
Iordanis Kerenidis and Anupam Prakash, \emph{Quantum recommendation systems}, 8th Innovations in Theoretical Computer Science Conference (ITCS 2017), pp.~49--1, Schloss Dagstuhl--Leibniz-Zentrum f{\"u}r Informatik, 2017, \arxiv{1603.08675}.

\bibitem{khachiyan1979polynomial}
Leonid~Genrikhovich Khachiyan, \emph{A polynomial algorithm in linear programming}, Doklady Akademii Nauk, vol. 244, pp.~1093--1096, Russian Academy of Sciences, 1979.

\bibitem{kothari2023mean}
Robin Kothari and Ryan O'Donnell, \emph{Mean estimation when you have the source code; or, quantum {M}onte {C}arlo methods}, Proceedings of the 2023 Annual ACM-SIAM Symposium on Discrete Algorithms (SODA), pp.~1186--1215, SIAM, 2023, \arxiv{2208.07544}.

\bibitem{legall2018improved}
Francois Le~Gall and Florent Urrutia, \emph{Improved rectangular matrix multiplication using powers of the {C}oppersmith-{W}inograd tensor}, Proceedings of the Twenty-Ninth Annual ACM-SIAM Symposium on Discrete Algorithms, pp.~1029--1046, SIAM, 2018.

\bibitem{lee2015efficient}
Yin~Tat Lee and Aaron Sidford, \emph{Efficient inverse maintenance and faster algorithms for linear programming}, 2015 IEEE 56th annual symposium on foundations of computer science, pp.~230--249, IEEE, 2015, \arxiv{1503.01752}.

\bibitem{lee2015faster}
Yin~Tat Lee, Aaron Sidford, and Sam Chiu-wai Wong, \emph{A faster cutting plane method and its implications for combinatorial and convex optimization}, 2015 IEEE 56th Annual Symposium on Foundations of Computer Science, pp.~1049--1065, IEEE, 2015, \arxiv{1508.04874}.

\bibitem{leng2023quantum}
Jiaqi Leng, Ethan Hickman, Joseph Li, and Xiaodi Wu, \emph{Quantum {H}amiltonian descent}, arXiv preprint arXiv:2303.01471 (2023), \arxiv{2303.01471}.

\bibitem{leng2025sub}
Jiaqi Leng, Kewen Wu, Xiaodi Wu, and Yufan Zheng, \emph{(sub) exponential quantum speedup for optimization}, arXiv preprint arXiv:2504.14841 (2025), \arxiv{2504.14841}.

\bibitem{li2021symmetry}
Xiaocheng Li, Chunlin Sun, and Yinyu Ye, \emph{The symmetry between arms and knapsacks: A primal-dual approach for bandits with knapsacks}, International Conference on Machine Learning, pp.~6483--6492, PMLR, 2021, \arxiv{2102.06385}.

\bibitem{liu2023quantum}
Yizhou Liu, Weijie~J. Su, and Tongyang Li, \emph{On quantum speedups for nonconvex optimization via quantum tunneling walks}, Quantum \textbf{7} (2023), 1030, \arxiv{2209.14501}.

\bibitem{megiddo1989varepsilon}
Nimrod Megiddo and Ramaswamy Chandrasekaran, \emph{On the $\varepsilon$-perturbation method for avoiding degeneracy}, Operations Research Letters \textbf{8} (1989), no.~6, 305--308.

\bibitem{mehta2007adwords}
Aranyak Mehta, Amin Saberi, Umesh Vazirani, and Vijay Vazirani, \emph{Adwords and generalized online matching}, Journal of the ACM (JACM) \textbf{54} (2007), no.~5, 22--es.

\bibitem{meyer2022survey}
Nico Meyer, Christian Ufrecht, Maniraman Periyasamy, Daniel~D. Scherer, Axel Plinge, and Christopher Mutschler, \emph{A survey on quantum reinforcement learning}, arXiv preprint arXiv:2211.03464 (2022), \arxiv{2211.03464}.

\bibitem{montanaro2015quantum}
Ashley Montanaro, \emph{Quantum speedup of {M}onte {C}arlo methods}, Proceedings of the Royal Society A: Mathematical, Physical and Engineering Sciences \textbf{471} (2015), no.~2181, 20150301, \arxiv{1504.06987}.

\bibitem{nannicini2024fast}
Giacomo Nannicini, \emph{Fast quantum subroutines for the simplex method}, Operations Research \textbf{72} (2024), no.~2, 763--780, \arxiv{1910.10649}.

\bibitem{robbins1952some}
Herbert Robbins, \emph{Some aspects of the sequential design of experiments}, Bulletin of the American Mathematical Society \textbf{58} (1952), no.~5, 527--535.

\bibitem{ruszczynski2010risk}
Andrzej Ruszczy{\'n}ski, \emph{Risk-averse dynamic programming for {M}arkov decision processes}, Mathematical programming \textbf{125} (2010), 235--261.

\bibitem{slivkins2019introduction}
Aleksandrs Slivkins, \emph{Introduction to multi-armed bandits}, Foundations and Trends{\textregistered} in Machine Learning \textbf{12} (2019), no.~1-2, 1--286, \arxiv{1904.07272}.

\bibitem{wan2023quantum}
Zongqi Wan, Zhijie Zhang, Tongyang Li, Jialin Zhang, and Xiaoming Sun, \emph{Quantum multi-armed bandits and stochastic linear bandits enjoy logarithmic regrets}, Proceedings of the AAAI Conference on Artificial Intelligence, vol.~37, pp.~10087--10094, 2023, \arxiv{2205.14988}.

\bibitem{wang2021quantum}
Daochen Wang, Xuchen You, Tongyang Li, and Andrew~M. Childs, \emph{Quantum exploration algorithms for multi-armed bandits}, Proceedings of the AAAI Conference on Artificial Intelligence, vol.~35, pp.~10102--10110, 2021, \arxiv{2007.07049}.

\bibitem{wang2020randomized}
Mengdi Wang, \emph{Randomized linear programming solves the {M}arkov decision problem in nearly linear (sometimes sublinear) time}, Mathematics of Operations Research \textbf{45} (2020), no.~2, 517--546.

\bibitem{weber1992gittins}
Richard Weber, \emph{On the {G}ittins index for multiarmed bandits}, The Annals of Applied Probability (1992), 1024--1033.

\bibitem{yu2009markov}
Jia~Yuan Yu, Shie Mannor, and Nahum Shimkin, \emph{Markov decision processes with arbitrary reward processes}, Mathematics of Operations Research \textbf{34} (2009), no.~3, 737--757.

\bibitem{zhang2023quantum}
Chenyi Zhang and Tongyang Li, \emph{Quantum lower bounds for finding stationary points of nonconvex functions}, International Conference on Machine Learning, pp.~41268--41299, PMLR, 2023, \arxiv{2212.03906}.

\bibitem{zhong2023provably}
Han Zhong, Jiachen Hu, Yecheng Xue, Tongyang Li, and Liwei Wang, \emph{Provably efficient exploration in quantum reinforcement learning with logarithmic worst-case regret}, arXiv preprint arXiv:2302.10796 (2023), \arxiv{2302.10796}.

\end{thebibliography}
\end{document}